\documentclass{ieeeaccess}
\usepackage{cite}
\usepackage{amsmath,amssymb,amsfonts}
\usepackage{graphicx}
\usepackage{textcomp}
\usepackage{hyperref}
\hypersetup{
  colorlinks,
  linkcolor={blue},
  citecolor={blue},
  urlcolor={blue}
}
\usepackage{extarrows}
\usepackage{algorithm}
\usepackage{algpseudocode}
\usepackage[caption=false]{subfig}
\usepackage{booktabs}

\usepackage{theorem}
\newtheorem{proposition}{Proposition}
\newtheorem{theorem}[proposition]{Theorem}
\newenvironment{proof}{\noindent \textbf{{Proof~} }}{\hfill $\square$}

\def\BibTeX{{\rm B\kern-.05em{\sc i\kern-.025em b}\kern-.08em
    T\kern-.1667em\lower.7ex\hbox{E}\kern-.125emX}}

\begin{document}

\history{Date of publication xxxx 00, 0000, date of current version xxxx 00, 0000.}
\doi{10.1109/TQE.2020.DOI}

\title{Improved Belief Propagation Decoding Algorithms for Surface Codes}
\author{\uppercase{Jiahan Chen}\authorrefmark{1},
\uppercase{Zhengzhong Yi\authorrefmark{1}, Zhipeng Liang\authorrefmark{1}, Xuan Wang}.\authorrefmark{1}
\IEEEmembership{Senior Member, IEEE}}
\address[1]{Harbin Institute of Technology, Shenzhen, China}

\tfootnote{This work was supported by Guangdong Provincial Key Laboratory of Novel Security Intelligence Technologies (No.2022B1212010005), Shenzhen Science and Technology Program, China (JCYJ20241202123906009), the Colleges and Universities Stable Support Project of Shenzhen, China (No.GXWD20220817164856008) and Harbin Institute of Technology, Shenzhen - SpinQ quantum information Joint Research Center Project (No.HITSZ20230111).}

\markboth
{Author \headeretal: Preparation of Papers for IEEE Transactions on Quantum Engineering}
{Author \headeretal: Preparation of Papers for IEEE Transactions on Quantum Engineering}

\corresp{Corresponding author: Zhengzhong Yi (email: zhengzhongyi@cs.hitsz.edu.cn), Xuan Wang (email: wangxuan@cs.hitsz.edu.cn).}

\begin{abstract}

Quantum error correction is crucial for universal fault-tolerant quantum computing. Highly accurate and low-time-complexity decoding algorithms play an indispensable role in ensuring quantum error correction works effectively. Among existing decoding algorithms, belief propagation (BP) is notable for its nearly linear time complexity and general applicability to stabilizer codes. However, BP's decoding accuracy without post-processing is unsatisfactory in most situations. This article focuses on improving the decoding accuracy of BP over GF(4) for surface codes. Inspired by machine learning optimization techniques, we first propose Momentum-BP and AdaGrad-BP to reduce oscillations in message updating, breaking the trapping sets of surface codes. We further propose EWAInit-BP, which adaptively updates initial probabilities and provides a 1 to 3 orders of magnitude improvement over traditional BP for planar surface code, toric code, and XZZX surface code without any post-processing method, showing high decoding accuracy even under parallel scheduling. The theoretical $O(1)$ time complexity under parallel implementation and high accuracy of EWAInit-BP make it a promising candidate for high-precision real-time decoders. 

\end{abstract}


\titlepgskip=-15pt

\maketitle

\section{Introduction}
\label{sec:introduction}
\PARstart{I}{nformation} in quantum computing devices is highly susceptible to noise, leading to information loss and computational errors \cite{wang_2017_coherence}. To address this issue, quantum error correction (QEC) introduces redundant physical qubits and encodes the information in a certain subspace of the state space of all the physical qubits to protect it. So far at least, QEC is an essential step towards universal fault-tolerant quantum computation \cite{zhao_2022_superconducting,google_2023_suppressing,shor_1996_FTQC}. To ensure QEC works, the decoder of a QEC code should output an accurate and fast estimation of error according to the error syndromes \cite{terhal_2015_rmp}. 

Several decoding algorithms have achieved notable success, including minimum-weight perfect matching (MWPM) \cite{edmonds_1965_MWPM} algorithm, union-find (UF) algorithm \cite{Nicolas_2022_UF}, maximum likelihood (ML) algorithm \cite{bravyi_2014_TN}, matrix-product-states-based (MPS-based) algorithm \cite{bravyi_2014_TN}, renormalization group algorithm \cite{duclos_2010_RG}, and neural-network-based (NN-based) algorithm \cite{Torlai_2017_neural1, breuckmann_2018_neural2, Liu_2019_Nerual3, ni_2020_neural4, Meinerz_2022_neural5, Overwater_2022_neural6, Choukroun_2024_neural7, Andreasson_2019_DRL}. Among these, ML, MWPM, and MPS-based algorithms offer high accuracy but come with high complexity. In contrast, the UF algorithm provides near-linear complexity with slightly lower accuracy. NN-based decoders have to balance scalability, decoding accuracy, and time cost \cite{Meinerz_2022_neural5}, and may struggle with generalizability across different codes. Consequently, existing decoding algorithms rarely meet the requirements of both high accuracy and low complexity in quantum error correction decoding \cite{holmes_2020_decodingspeed, battistel_2023_realtimedecoding}.

Belief Propagation (BP) algorithm \cite{Mackay_2004_quantumBP} are applicable to any stabilizer code and can achieve a complexity of $O(j\tau N)$ \cite{kuo_2021_quantumBP9_MBP_gf4}, where $j$ is the average weight of the rows and columns of the check matrix, and $\tau$ is the number of iterations. Additionally, parallel scheduling of BP can be efficiently implemented in hardware, theoretically achieving constant-time complexity \cite{valls_2021_fpga}. However, traditional BP exhibits low decoding accuracy for many quantum error correcting codes and even cannot achieve the code capacity threshold for surface codes, which are the focus of recent advancements in quantum error correction implementation \cite{google_2023_suppressing, google_2024_surface}.

BP can be implemented over both GF(4) \cite{poulin_2008_quantumBP1_gf4, wang_2012_quantumBP2_gf4, babar_2015_quantumBP3_gf4, raveendran_2019_quantumBP4_GBP2_gf4, kuo_2020_quantumBP6_gf4, kuo_2021_quantumBP9_MBP_gf4, Kuo_2022_quantumBP10_gf4} and GF(2) \cite{rigby_2019_quantumBP5_gf2, Lai_2021_quantumBP8_gf2, yi_2023_quantumBP11_gf2, old_2022_quantumBP12_GBP_gf2}, where the Pauli operators are represented. There have been some efforts improving decoding accuracy of BP over GF(4) and GF(2). Over GF(4), MBP \cite{kuo_2021_quantumBP9_MBP_gf4} achieves a threshold of 14.5\% to 16\% for the planar surface code by modifying the step size in the \textit{a posteriori} probability updates. However, MBP is limited to serial scheduling, and the outer loop for searching step sizes increases its time complexity. Over GF(2), GBP \cite{old_2022_quantumBP12_GBP_gf2} achieves a 14\% threshold by combining generalized belief propagation with multiple rounds of decoding. This algorithm can be implemented in parallel, but the number of outer loops increases with the code length, adding to its time complexity. BP-OTS \cite{Chytas_2024_quantumts3_BPOTS} periodically detects and breaks trapping sets using manually set parameters, allowing parallel implementation while maintaining the original complexity. It outperforms MBP for short codes but does not achieve a significant threshold.

Post-processing methods, such as OSD, SI (stabilizer inactivation) \cite{Crest_2022_quantumslowBP5_SI}, MWPM \cite{criger_2018_quantumslowBP1_MWPM, higgott_2022_quantumslowBP4_beliefmatching, caune_2023_quantumslowBP6_MWPM} and UN\cite{google_2023_suppressing} are often used to improve BP's decoding accuracy. However, they also introduce extra time cost and all have high time complexities, ranging from $O(N^2)$ to $O(N^3)$. Improving BP's decoding accuracy to get rid of its reliance on post-processing algorithms remains a significant challenge.

In this article, we mainly focus on improving BP's decoding accuracy for surface codes, while we also believe that the ideas used in these improvements will also work for some other quantum LDPC codes. There are several explanations for why surface codes are challenging for BP decoding. From the perspective of degeneracy \cite{fuentes_2021_degeneracy, kuo_2021_quantumBP9_MBP_gf4}, stabilizers of surface codes have much lower weights compared to the code distance, leading to multiple low-weight errors corresponding to the same syndrome, thus having negative influence on BP's error convergence. From the perspective of short cycles \cite{Mackay_2004_quantumBP}, a number of 4-cycles over GF(4) and 8-cycles over GF(2) in surface codes can cause BP's inaccurate updating. A more intuitive explanation involves trapping sets \cite{raveendran_2021_quantumts1, Chytas_2024_quantumts3_BPOTS}, where surface codes have many symmetric trapping sets causing BP messages to oscillate periodically.

This paper approaches BP improvements from two perspectives: optimizing message updates inspired by gradient optimization, and optimizing \textit{a priori} probabilities using information of previous iterations, leading to the proposal of three algorithms.

First, inspired by the similarity between BP message update and gradient descent of the energy function, we leverage optimization methods in machine learning to optimize the message update rules, and propose Momentum-BP and AdaGrad-BP. In addition to improving decoding accuracy, experiments on trapping sets in surface codes show that Momentum-BP and AdaGrad-BP can smooth the message update process and reduce oscillation amplitude.

Furthermore, realizing the importance of \textit{a priori} probabilities involved in the iteration of BP, we propose EWAInit-BP which dynamically updates the \textit{a priori} probabilities involved in each iteration by utilizing the \textit{a posteriori} probabilities output by previous iterations. 
Experiments on short toric codes, planar surface codes, and XZZX surface codes (under biased noise) show that EWAInit-BP outperforms existing BP improvements without post-processing and outer loops under parallel scheduling. Notably, the time complexity of all proposed improvements remains $O(j\tau N)$. Thus, EWAInit-BP can be implemented with theoretically constant complexity in parallel while maintaining decoding accuracy.

The structure of this article is as follows. Section \ref{preliminaries} introduces the fundamentals of surface codes and BP decoding, as well as the perspective of BP as a gradient optimization, and summarizes the challenges BP faces in decoding surface codes. Section \ref{belief propagation optimization} introduces Momentum-BP, AdaGrad-BP and EWAInit-BP. The results of the simulation are presented in Section \ref{simulation results}. Section \ref{summary} summarizes our work and discusses the future challenges.

\section{Preliminaries}
\label{preliminaries}

\subsection{Surface codes}

A stabilizer code is defined by an abelian subgroup $ \mathcal{S} $ of the Pauli group $ \mathcal{P}_n $ excluding $-I$, and composed of operators $ S \in \mathcal{P}_n $ such that $ S|\psi\rangle = |\psi\rangle $ for all codewords $ |\psi\rangle $ in the code space. The code space is the simultaneous $ +1 $ eigenspace of all elements of $ \mathcal{S} $. The stabilizers must commute with each other, which means $ [S_i, S_j] = 0 $ for  $ \forall S_i, S_j \in \mathcal{S} $. The dimension of the code space is determined by the number of independent generators of $ \mathcal{S} $. For a stabilizer code with $ N $ physical qubits and $ K $ logical qubits, $ \mathcal{S} $ has $ N-K $ independent generators.

The stabilizers are used to detect errors by measuring their eigenvalues. The set of measurement outcomes, known as the \textit{error syndrome}, contains information about the errors on physical qubits. The logical operators of a stabilizer code are Pauli operators that commute with all elements of $ \mathcal{S} $ but are not in $ \mathcal{S} $, thus they are in $\mathcal{N(S)\backslash S}$ where $\mathcal{N(S)} = \{E: ES=SE, \ \forall S \in \mathcal{S}\}$. They act on the encoded logical qubits without disturbing the code space, thus they can lead to \textit{undetectable errors}. The code distance $ D $ is defined as the minimum weight of logical operators. A quantum stabilizer code can be represented by $[[N, K, D]]$.

Surface codes may be the most notable stabilizer codes in recent years. Surface codes are a class of quantum Low-Density Parity-Check (LDPC) codes \cite{breuckmann_2021_QLDPC} where qubits are placed on a two-dimensional lattice, relying on the topology of the surface to protect quantum information \cite{bombin_2013_introductiontopological}. Logical qubits are encoded in global degrees of freedom, and larger lattices provide larger code distances. The code distance scales with the code length as $O(\sqrt{N})$, whereas the code rate is asymptotically zero. In this article, toric codes refer to those with closed boundaries, and planar surface codes refer to those with open boundaries. Besides, We consider the unrotated surface codes, but it is clear that our proposed algorithms can be extended to rotated versions.

A $[[2L^2, 2, L]]$ toric code has a 2-dimensional torus topology, and the closed boundaries makes it translationally invariant to the error syndrome, which is beneficial to decoding algorithms such as BP and neural networks. The Euler characteristic of the torus is 0, but the X-type and Z-type stabilizers of the toric code both have a redundancy, resulting in a total of two encoded logical qubits.

A $[[2L^2-2L+1, 1, L]]$ planar surface code is defined on a surface with open boundaries, which makes the stabilizers near the boundaries having lower weights. This reduction in stabilizer weight results in a slight decrease in the error correction capability compared to the toric code; however, it simplifies implementation in physical quantum systems.. The lattice structure of the planar surface codes are illustrated in Fig.~\ref{fig:surface_topo}.

A $[[L^2, 2, L]]$ XZZX surface code is a non-CSS code designed to handle biased noise more effectively. The stabilizers of these codes have the form of $XZZX$. This alternating pattern of Pauli $X$ and $Z$ operators allows the code to reach the theoretical maximum threshold of 50\% under pure Pauli $X$, pure Pauli $Y$ and pure Pauli $Z$ noise.

\Figure[t!](topskip=-10pt, botskip=0pt, midskip=0pt)[width=2.5in]{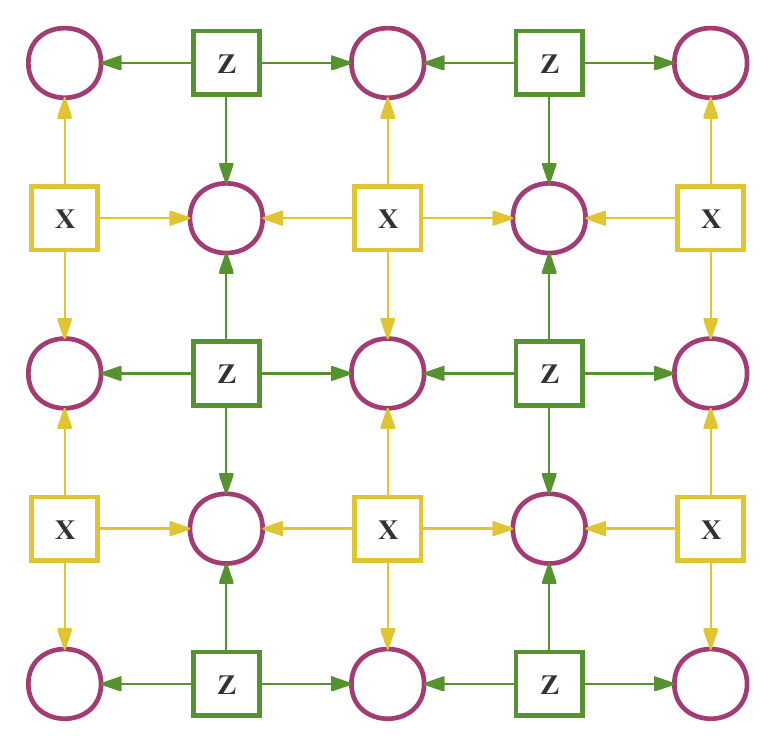}
{The lattice of the $[[2L^2-2L+1,1,L]]$ planar surface code with L=3. Circles represent data qubits, yellow squares represent X-type stabilizers, green squares represent Z-type stabilizers, and arrows indicate the coupling between stabilizers and qubits.\label{fig:surface_topo}}

\subsection{Belief propagation decoding}

\subsubsection{BP over GF(4)}

The idea behind \textit{maximum likelihood decoding} in quantum error correction is to find the most likely error coset given the error syndrome. This can be expressed as \cite{poulin_2008_quantumBP1_gf4}
\begin{align}
    E = \arg\max_{E\mathcal{S}} \sum_{S\in \mathcal{S}} p(ES|s) \label{qecobject}
\end{align}
where the term $ ES $ denotes the coset formed by the error $ E $ and the stabilizer $ S $. However, finding this maximum likelihood solution is computationally expensive and exponential \cite{hsieh_2011_hardness}. Belief propagation decoding algorithm simplifies the above objective to
\begin{align}
    E = \left\{\arg\max_{E_j} p(E_j|s)\right\}_{j=1}^{n}  \label{bpobject}
\end{align}
where the probability of each individual error $ E_j $ given the syndrome $ s $ is maximized independently for $ j = 1 $ to $ n $. Belief propagation makes two simplifications to the objective, which compromise its performance on highly degenerate quantum codes.
\begin{itemize}
    \item Assuming the most likely error pattern is the same as the most likely error coset;
    \item Assuming the most likely error on each qubit is the same as the most likely error pattern.
\end{itemize}

BP algorithm passes messages on a \textit{Tanner graph}, where stabilizers and qubits are represented as \textit{check nodes} and \textit{variable nodes}, respectively. Over GF(4) which this article mainly focuses on, each variable node corresponds to a probability vector of all Pauli errors on one qubit. 

Besides, one can use log-likelihood ratio (LLR) as the message passed between variable nodes and check nodes. In this way, belief propagation (abbreviated as LLR-BP\textsubscript{4}) receives \textit{a priori} probability vectors describing each type of Pauli error on each qubit as $\{p_{v_j}=\{p_{v_j}^X, p_{v_j}^Y, p_{v_j}^Z\}\}_{j=1}^N$, and initializes the variable-to-check messages as \cite{kuo_2021_quantumBP9_MBP_gf4}
\begin{align}
    \Pi_{v_j} &= \{{\Pi_{v_j, W}}\}_{W \in \{X,Y,Z\}} = \ln \frac{1-p_{v_j}}{p_{v_j}}, \\
    m_{v_j\to c_i}^{(0)} &= \lambda_{H_{ij}} ({\Pi_{v_j, X}}, {\Pi_{v_j, Y}}, {\Pi_{v_j, Z}}), \label{init}
\end{align}
where $\Pi_{v_j}$ is the \textit{a priori} probabilities of variable node $j$ expressed by LLR. The lambda function in \eqref{init} is defined as
\begin{align}
    \lambda_{H_{ij}} (p^X, p^Y, p^Z) = \ln \frac{1+e^{-p^{H_{ij}}}}{e^{-p^X}+e^{-p^Y}+e^{-p^Z}-e^{-p^{H_{ij}}}}
\end{align}
indicating LLR values of whether the Pauli error $ E_j $ commutes with the Pauli operator $H_{ij}$ of the parity-check matrix $H$. $H_{ij}$ is also the $j$-th element of stabilizer $ S_i $.

One iteration of BP consists of one horizontal update, one vertical update and one hard decision. In the horizontal update, each check node calculates the feedback based on the messages received from the variable nodes but not from the target variable node in last vertical update, sending a check-to-variable message as follows
\begin{align}
    w = 2\tanh^{-1} \left( \prod_{v_{j'} \in \mathcal{N}(c_i) \backslash v_j} \tanh \left(\frac{m_{v_{j'} \to c_i}^{(t-1)}}{2}\right) \right) \label{h_update},
\end{align}
\begin{align}
    m_{c_i \to v_j}^{(t)} = (-1)^{s_i} * w,
\end{align}
where $\mathcal{N}(c_i)$ represents all variable nodes adjacent to check node $i$, and $s_i$ is the syndrome corresponding to stabilizer $S_i$.

In the vertical update, variable nodes aggregate the messages from check nodes and, after a similar extrinsic calculation as above, send variable-to-check messages as
\begin{align}
    m_{v_j \to c_i}^{(t)} = \lambda_{H_{ij}} (\{\Pi_{v_j, W} + \sum_{\substack{c_{i'} \in \mathcal{M}(v_j) \backslash c_i \\ \left\langle W, H_{ij}\right\rangle=1}} m_{c_{i'} \rightarrow v_j}^{(t)}\}_{W \in \{X,Y,Z\}}) \label{v_update},
\end{align}
where $\mathcal{M}(v_j)$ represents all check nodes adjacent to variable node $j$. $\left\langle W, H_{ij} \right\rangle$ stands for the trace inner product, here the result is 1 when the error operator $W$ anti-commutes with the matrix element $H_{ij}$ and 0 when they commute.

In the hard decision, variable nodes first accumulate the messages from all check nodes to calculate the \textit{a posteriori} probability of each Pauli error on each qubit as
\begin{align}
    Q_{v_j, W}^{(t)} = \Pi_{v_j, W} + \sum_{\substack{c_i \in \mathcal{M}(v_j) \\ \left\langle W, H_{ij}\right\rangle=1}} m_{c_i \rightarrow v_j}^{(t)}. \label{p_update}
\end{align}
Then, for each variable node, if all log-likelihood ratios for Pauli errors are greater than 0, then the hard decision result is the identity $I$; otherwise, it is the Pauli error with the smallest LLR value. The hard decisions of all qubits combine to form the error estimate $\hat{E}$. If $\hat{E}$ corrects the error syndrome, the algorithm converges; otherwise, it returns to the horizontal update for a new iteration. If the maximum number of iterations is reached, decoding has failed.

BP can be performed under either parallel or serial scheduling \cite{hocevar_2004_LayerScheduling}. Under parallel scheduling, all variable nodes and check nodes update their messages simultaneously in each iteration, allowing for hardware parallelization and constant time complexity. On the other hand, the serial scheduling updates the messages of variable nodes and check nodes sequentially. This method leads to faster convergence and slightly higher accuracy, though it has a more complex hardware implementation and higher computational overhead.

One perspective on BP is that each \textit{a posteriori} probabilities update in one BP iteration can be viewed as a single step of gradient descent with a simplified step size on a certain energy function, as the expanded form of the \textit{a posteriori} update closely resembles the derivative of the energy function (see Appendix \ref{app:gradient}). From this perspective, MBP and AMBP \cite{kuo_2021_quantumBP9_MBP_gf4} adjusts the "learning rate" of message updates to converge more quickly to the equivalent energy minimum, which is common in degenerate quantum codes.

\subsubsection{Hardness of BP in decoding surface codes}

There are several terms that can explain the unsatisfactory accuracy of BP in surface codes.

\paragraph{Cycles} 
Short cycles cause the joint probability calculations in BP to be dependent, which violates the independence assumption of updating variable nodes \eqref{v_update} and check nodes \eqref{h_update} \cite{Mackay_2004_quantumBP}. Surface codes exhibit 8-cycles over GF(2) and 4-cycles over GF(4), exacerbating this issue.

\paragraph{Degeneracy} 
High degeneracy leads to BP decoding objectives \eqref{bpobject} deviating from true targets \eqref{qecobject} \cite{poulin_2008_quantumBP1_gf4, fuentes_2021_degeneracy}, especially when errors can combine with stabilizers to form error types that are symmetric to the original ones. In such cases, BP's marginal probability calculations for the both error types will be completely consistent. In surface codes, the weights of the stabilizers are all no greater than 4, often much less than the weights of the logical operators. This symmetric degeneracy results in incorrect convergence of BP decoding.

\paragraph{Trapping Sets}
Trapping sets are specific local structures in error-correcting codes combined with corresponding error patterns, where iterative decoding algorithms can never converge \cite{dehkordi_2010_classicts1,han_2022_classicts2,kang_2016_classicts3}. In QEC, trapping sets are caused by the aforementioned short loops and symmetric degeneracy \cite{raveendran_2021_quantumts1, Pradhan_2023_quantumts2}. The smallest trapping set in surface codes is shown in Fig.~\ref{fig:ts}, consisting of an 8-cycle and four symmetric stabilizers, where the local connections of each check node and variable node are identical \cite{Chytas_2024_quantumts3_BPOTS}. If any two qubits here experience errors that anticommute with the stabilizers, the hard decision result of BP will oscillate periodically between all 0s and all 1s.

The first two terms are directly related to the underlying principles of the BP algorithm. In contrast, trapping sets offer a more specific description, which can not only inspire potential improvements but also serve as a metric for evaluating improvements.

\begin{figure}[t]
  \centering
  \subfloat[]{\includegraphics[width=0.465\columnwidth]{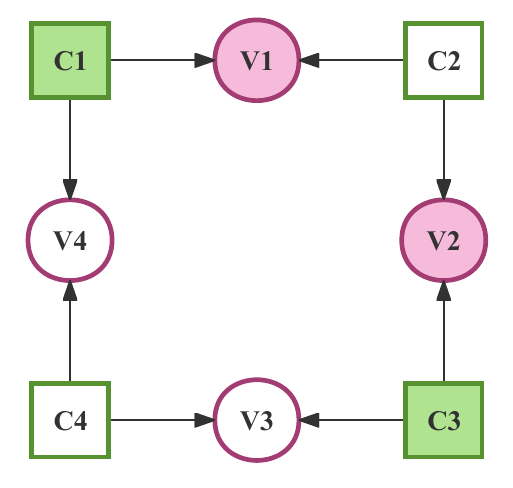}\label{fig:ts_1}}
  \hfill
  \subfloat[]{\includegraphics[width=0.45\columnwidth]{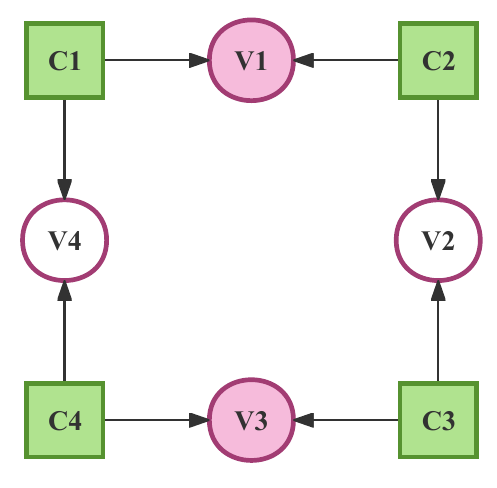}\label{fig:ts_2}}

  \caption{Trapping sets with two types of error patterns in surface codes. Solid circles represent error qubits, and solid squares represent stabilizers with non-trivial error syndromes. The stabilizers in the figure can be either all X-type or all Z-type.}
  \label{fig:ts}
\end{figure}

\section{Belief Propagation Optimization}
\label{belief propagation optimization}

From the perspective of BP as gradient descent \cite{kuo_2021_quantumBP9_MBP_gf4}, the high degeneracy of surface codes results in many equivalent global minima, leading to BP oscillating between adjacent erroneous solutions within symmetric trapping sets. A natural idea is to explore if certain empirical methods can be employed to smooth BP’s message updates, stabilizing the \textit{a posteriori} probabilities at a specific local minimum before oscillations occur, thus enabling successful convergence to degenerate errors.

\subsection{Gradient Optimization of Message updates}
\label{Message update optimization}

This subsection follows the gradient perspective, inspired by optimizers in machine learning, to smooth BP iterations by modifying its message update rules.

\subsubsection{Update gradient smoothing} 

The Exponential Weighted Average (EWA) method updates and smooths data series by combining the current observation with a weighted average of the previous value, expressed as
\begin{align}
  \label{EWA}
    v_t = \gamma v_{t-1} + (1 - \gamma) x_t,
\end{align}
where $ x_t $ is the observation at time $ t $, $ v_t $ is the smoothed value, and $ \gamma $ is a decay factor satisfying $0 < \gamma < 1$.

In machine learning, momentum methods utilize an idea similar to EWA, adding a momentum term to accumulate historical information in gradient descent as
\begin{align}
  \label{Momentum}
    m_t &= \gamma m_{t-1} + (1-\gamma) \nabla J(\theta_{t-1}), \\
    \theta_t &= \theta_{t-1} - \alpha m_t,
\end{align}
where $ \gamma $ is the decay factor similar to that in EWA, $ m_t $ is the accumulated momentum term, $ \nabla J(\theta) $ is the gradient of the loss function with respect to parameter $ \theta $, and $ \alpha $ is the learning rate. This method integrates previous directions while updating parameters, thereby smoothing the update direction.

We incorporate the concept of momentum into BP's \textit{a posteriori} probability updates to solve this issue by breaking the symmetry in BP's message updates, proposing Momentum-BP. Momentum-BP can be implemented under both parallel and serial scheduling. Here, we present the latter in Algorithm~\ref{algo:momentumBP}, where the initialization and check node update remain the same as in traditional BP. Our approach computes the \textit{a posteriori} probability before performs the vertical update. The posterior update rule of Momentum-BP is shown as \eqref{mombp1} and \eqref{mombp2} where $ g_{v_j,W}^{(t)} $ are the accumulated momentum term and $ Q_{v_j,W}^{(t)} $ are the \textit{a posteriori} probability for the $ j $-th variable node regarding the Pauli error $ W $ in the $ t $-th iteration. Notice that we transformed \eqref{p_update}, so that the accumulated term is the posterior value from the previous iteration. In the variable node update step, we subtract the original message sent from the check node rather than the momentum-smoothed message to ensure stability and prevent overcompensation.

For implementation of the gradient calculation, we ignore the coefficient terms in \eqref{partial}, using the difference between the posterior updates calculated by traditional BP and the \textit{a posteriori} probabilities from the previous iteration as the gradient for the current iteration, expanded in \eqref{mombp3}. When $ \alpha=1 $ and $ \gamma=0 $, \eqref{mombp3} is equivalent to the traditional BP's \textit{a posteriori} update; when $ 0 < \alpha < 1 $ and $ \gamma=0 $, it is equivalent to the EWA of \textit{a posteriori} updates. In practice, to avoid the malice of hyper-parameter tuning, we either fix $ \alpha $ at 1 or $ \gamma $ at 0.

The properties of Momentum-BP can be explained from two perspectives. (1) Since it takes into account message updates from previous iterations, the method smooths the update direction in each iteration, preventing the corresponding energy function from oscillating between two points and gradually converging to the global minimum. (2) Node-wise momentum updates gradually break the symmetry of message passing shown in Fig.~\ref{fig:ts}, allowing BP to escape from these trapping sets. The trapping set simulation in Section \ref{Breaking the Trapping Sets} will support those explanations.

\begin{algorithm}[H]
    \caption{Momentum-BP under Serial Scheduling}
    \label{algo:momentumBP}
    
    \textbf{Input:} Parity-check matrix $\boldsymbol H$ of $M \times N$, syndrome vector $\boldsymbol s = (1, -1)^M$, \textit{a priori} LLR vectors $\{\Pi_{v_j}^{(0)}\}_{j=1}^N$, max iterations $iter_{max}$, learning rate $\alpha=1$, and momentum rate $\gamma \geq 0$.

    \textbf{Output:} Estimated error vector $\boldsymbol {\hat E}$.

    \textbf{Initialization:} \\
        $t=1$ \\
        \textbf{for} $j=1\to N$, \textbf{do}
        
        \qquad Initialize $m_{v_j\to c_i}^{(0)}$ as \eqref{init}, let $g_{v_j}^{(0)} = 0, \ Q_{v_j}^{(0)} = \Pi_{v_j}$
    
    \textbf{Serial update:} \\
        \textbf{for} $j=1\to N$, $W \in \{X,Y,Z\}$, \textbf{do}
        
            \qquad \textbf{for} $i \in \mathcal{N}(v_j)$, \textbf{do}
            
                \qquad \qquad Calculate $m_{c_i \rightarrow v_j}^{(t)}$ as \eqref{h_update}

        \qquad \textbf{Momentum update}
        \begin{align}
        	g_{v_j, W}^{(t)} &= \gamma g_{v_j, W}^{(t-1)} + (1-\gamma)\frac{\partial J}{\partial Q_{v_j, W}^{(t-1)}} \label{mombp1} \\
        	Q_{v_j, W}^{(t)} &= Q_{v_j, W}^{(t-1)} - \alpha g_{v_j, W}^{(t)} \label{mombp2}
        \end{align}
        \qquad Variable node update
        \begin{align} \label{mombp5}
            m_{v_j \to c_i}^{(t)} = Q_{v_j, W}^{(t)} - \langle W, H_{ij} \rangle m_{c_i \to v_j}^{(t)}
        \end{align}
        
    \textbf{Hard decision:}
    
        \textbf{for} $j=1\to N$, $W \in \{X,Y,Z\}$, \textbf{do}
        
            \qquad \textbf{if} $Q_{v_j, W}^{(t)} > 0$ for all $W$, \textbf{let} $\hat E_i=I$
        
            \qquad \textbf{else, let} $\hat E_i=\underset{W \in\{X, Y, Z\}}{\arg \min } Q_{v_j,W}^{(t)}$

        \textbf{if} $\langle H, \hat E \rangle = s$, \textbf{return} "converge" and halt. \\
        \textbf{else if} $t = iter_{max}$, \textbf{return} "fail" and halt. \\
        \textbf{else} $t=t+1$, repeat from \textbf{serial update}.
\end{algorithm}

\begin{figure*}[t]
\centering
\begin{align}
    Q_{v_j,W}^{(t)} &= Q_{v_j,W}^{(t-1)} - \alpha \left( \gamma g_{v_j,W}^{(t-1)} + (1-\gamma) (Q_{v_j,W}^{(t-1)} - \Pi_{v_j,W} - \sum_{\substack{c_i \in \mathcal{M}(v_j) \\ \langle W, H_{ij}\rangle=1}} m_{c_i \rightarrow v_j}^{(t)}) \right) \label{mombp3} \\
    & \xlongequal{\text{if } \alpha=1} \gamma (Q_{v_j,W}^{(t-1)} - g_{v_j,W}^{(t-1)}) + (1-\gamma)(\Pi_{v_j,W} + \sum_{\substack{c_i \in \mathcal{M}(v_j) \\ \langle W, H_{ij}\rangle=1}} m_{c_i \rightarrow v_j}^{(t)}). \\
    & \xlongequal{\text{if } \gamma=0} \alpha (\Pi_{v_j,W} + \sum_{\substack{c_i \in \mathcal{M}(v_j) \\ \langle W, H_{ij}\rangle=1}} m_{c_i \rightarrow v_j}^{(t)}) + (1-\alpha) Q_{v_j,W}^{(t-1)}. 
    \label{mombp4}
\end{align}
\end{figure*}

\subsubsection{Adaptive step size}

Oscillation of message updates in trapping sets can also be attributed to the step size. In traditional LLR-BP, all variable node posterior updates result from simply directly accumulating messages from neighboring check nodes. From the gradient perspective, if the step size is too large, the optimization algorithm is likely to oscillate; if the step size is too small, the algorithm is likely unable to update effectively. In practical decoding, the update frequency of different variable nodes may be uneven, with some requiring more rapid updates while others need to remain stable.

Inspired by adaptive gradient methods in deep learning, we propose AdaGrad-BP. The framework of this algorithm is similar to Momentum-BP, with initialization and horizontal updates identical to traditional BP, while the posterior probability updates follow an adaptive step size rule as
\begin{align}
    Q_{v_j,W}^{(t)} = Q_{v_j,W}^{(t-1)} - \alpha \frac{{g^{'}}_{v_j,W}^{(t)}}{\sqrt{G_{v_j,W}^{(t)}}+\epsilon}, \label{boost}
\end{align}
where $ {g^{'}}_{v_j,W}^{(t)} $ approximates the gradient of energy function
\begin{align}
    {g^{'}}_{v_j,W}^{(t)} =  Q_{v_j,W}^{(t-1)} - \Pi_{v_j,W} - \sum_{\substack{c_i \in \mathcal{M}(v_j) \\\left\langle W, H_{ij}\right\rangle=1}} m_{c_i \rightarrow v_j}^{(t)},
\end{align}
and $ G_{v_j,W}^{(t)} $ is the cumulative sum of the squared gradients as
\begin{align}
    G_{v_j,W}^{(t)} = G_{v_j,W}^{(t-1)} + ({g^{'}}_{v_j,W}^{(t)})^2.
\end{align}

This method adaptively adjusts the step size by accumulating historical gradient information. If the estimate of a variable node for a certain Pauli error changes significantly, the algorithm reduces the step size to counteract oscillation. \eqref{boost} takes effect after the first iteration to ensure initial update stability. Here, $\alpha$ is the boost learning rate, providing an update impetus as the algorithm begins to accumulate gradients. $\alpha$ is relatively insensitive, so it is fixed at 5 in our simulations. While AdaGrad-BP may face issues with accumulated gradients causing the step size to approach zero after many iterations, in practice it can reduce the number of iterations, allowing most error patterns to converge before the gradients become too small.

Unfortunately, combining Momentum-BP and AdaGrad-BP using Adam optimizer or other methods will not yield satisfactory results. On one hand, the Adam algorithm introduces additional hyperparameters whose optimal values for BP cannot be directly borrowed from deep learning, thereby complicating the algorithm. On the other hand, Momentum and AdaGrad approach similar message update optimizations from different perspectives, and a straightforward combination of the two can lead to overcompensation.

\subsection{Dynamic Initialization of Prior Probabilities}
\label{Dynamic Initialization}

BP decoding and gradient descent still have a fundamental difference, which limits the performance of the algorithms in Section \ref{Message update optimization}. Each update in gradient descent is based on the results of the previous step, whereas in BP decoding, each calculation of \textit{a posteriori} probabilities accumulates the same \textit{a priori} probabilities. Given the objective function as shown in \eqref{bpobject}, for simplicity using the binary domain for instance, BP calculates the LLR values as
\begin{align}
    Q_{v_j} &= \ln \frac{P(e_j=0|s)}{P(e_j=1|s)} \\
    &= \ln \frac{P(e_j=0)\cdot P(s|e_j=0)}{P(e_j=1)\cdot P(s|e_j=1)} \\
    &= \Pi_{v_j} + \sum_{\substack{c_{i'} \in \mathcal{M}(v_j) \backslash c_i}} m_{c_{i'} \rightarrow v_j}, 
\end{align}
where $ \Pi_{v_j} $ is a manually set initial value of the algorithm that participates in each iteration's \textit{a posteriori} probability update. However, since the initial error probabilities are derived from a theoretical channel which is difficult to measure and estimate in experiments, the \textit{a priori} probabilities can significantly affect the BP iterations.

Existing methods for altering initial probabilities include reinitialization using information obtained from previous iterations \cite{poulin_2008_quantumBP1_gf4, wang_2012_quantumBP2_gf4, old_2022_quantumBP12_GBP_gf2}, as well as detecting oscillations during the decoding process and modifying initial values accordingly \cite{Chytas_2024_quantumts3_BPOTS}. However, they either require multiple outer loops or manually set algorithm parameters. We introduce a simple but effective method for dynamically updating \textit{a priori} probabilities, employing a transformed exponential weighted average, incorporating components of previous \textit{a posteriori} probabilities into the \textit{a priori} probabilities with a decaying factor. Hence, we call this method EWAInit-BP. EWAInit-BP under parallel scheduling is shown in Algorithm~\ref{algo:EWAInitBP}.

One might wonder why the first term on the RHS of \eqref{EWAInit} is each qubit's list of initial \textit{a priori} probabilities $ \Pi_{v_j}^{(0)} $ instead of \textit{a priori} probabilities from the previous iteration $ \Pi_{v_j}^{(t-1)} $, as the latter aligns better with the definition and intuition of the exponential weighted average. In fact, by using the definition in \eqref{EWAInit} rather than the aforementioned approach, we achieve the exponentially weighted average of horizontal message updates for the \textit{a posteriori} probabilities in EWAInit-BP, as stated in Theorem \ref{theorem1}.

\begin{algorithm}[H]
    \caption{EWAInit-BP under Parallel Scheduling}
    \label{algo:EWAInitBP}
    
    \textbf{Input:} Parity-check matrix $\boldsymbol H$, syndrome vector $\boldsymbol s$, \textit{a priori} LLR vectors $\{\Pi_{v_j}^{(0)}\}_{j=1}^N$, max iterations $iter_{max}$, discount factor $\alpha$.

    \textbf{Output:} Estimated error vector $\boldsymbol {\hat E}$.

    $t=1, \ \Pi^{(1)}=\Pi^{(0)}, \ m_{v_j\to c_i}^{(0)} = \lambda_{H_{ij}} (\Pi_{v_j, X}^{(0)}, \Pi_{v_j, Y}^{(0)}, \Pi_{v_j, Z}^{(0)})$ \\
    \textbf{EWA Initialization when $t > 1$:} \\
        \textbf{for} $j=1\to N$, \textbf{do}
        \begin{align} \label{EWAInit}
            \Pi_{v_j}^{(t)} = \alpha \Pi_{v_j}^{(0)} + (1-\alpha) Q_{v_j}^{(t-1)}
        \end{align}

    \textbf{Horizontal update:} \\
        \textbf{for} $i=1\to M$, Calculate $m_{c_i \rightarrow v_j}^{(t)}$ as \eqref{h_update}

    \textbf{Vertical update:} \\
        \textbf{for} $j=1\to N$,
        
        \qquad Calculate $m_{v_j \rightarrow c_i}^{(t)}$ as \eqref{v_update}, replacing $\Pi_{v_j, W}$ with $\Pi_{v_j, W}^{(t)}$
        
    \textbf{Hard decision:} \\
        \textbf{for} $j=1\to N$, $W \in \{X,Y,Z\}$, \textbf{do}
        
        \qquad Calculate $Q_{v_j, W}^{(t)}$ as \eqref{p_update}, replacing $\Pi_{v_j, W}$ with $\Pi_{v_j, W}^{(t)}$
    
        \qquad \textbf{if} $Q_{v_j,W}^{(t)} > 0$ for all $e$, \textbf{let} $\hat E_i=I$
        
        \qquad \textbf{else, let} $\hat E_i=\underset{W \in\{X, Y, Z\}}{\arg \min } Q_{v_j,W}^{(t)}$

        \textbf{if} $\langle H, \hat E \rangle = s$, \textbf{return} "converge" and halt. \\
        \textbf{else if} $t = iter_{max}$, \textbf{return} "fail" and halt. \\
        \textbf{else} $t=t+1$, repeat from \textbf{EWA Initialization}.
\end{algorithm}

\begin{theorem}
\label{theorem1}
The a posteriori update of EWAInit-BP (Algorithm \ref{algo:EWAInitBP}) realizes an exponentially weighted average of past horizontal message updates, specifically expressed as
\begin{align}
    Q_{v_j, W}^{(t)} &= \Pi_{v_j, W}^{(0)} + \sum_{k=0}^{t} (1 - \alpha)^k \sum_{\substack{c_i \in \mathcal{M}(v_j) \\ \langle W, H_{ij} \rangle = 1}} m_{c_i \rightarrow v_j}^{(t-k)}.
\end{align}
\end{theorem}

\begin{proof}
This result can be derived by combining \eqref{EWAInit} and \eqref{p_update} and recursively expanding the update formula. The detailed process is provided in Appendix \ref{app:expansion}.
\end{proof}

Notably, EWAInit-BP and Momentum-BP share a similar recusive form, but the former is more concise (see Appendix \ref{app:expansion}). The primary distinction between them lies in the differing placements of EWA (or momentum). Moreover, by incorporating the EWA step into the initialization phase, EWAInit-BP avoids the issue of message coefficient inconsistency caused by performing posterior updates followed by vertical updates. In fact, EWAInit-BP also demonstrates better performance in simulations.

\subsection{Complexity}

Our proposed algorithms involve the same message passing between variable nodes and check nodes as in traditional BP. The main difference lies in the additional update steps, which in Momentum-BP and AdaGrad-BP consist of recalculating the \textit{a posteriori} probabilities, and in EWAInit-BP involves recalculating the \textit{a priori} probabilities. 
These additional steps only involve simple arithmetic operations and can be executed locally at each node, without requiring extra iterations, loops, or global communication, thus not introducing any significant overhead. Therefore, the overall time complexity of the proposed algorithms remains $O(j\tau N)$, where $j$ is the average degree of the Tanner graph, and $\tau$ is the number of iterations.

In terms of space complexity, traditional BP requires storing messages and probability values, resulting in $O(N)$ space requirements. Our algorithms only add extra storage space to accumulate gradients or maintain exponential weighted averages which is proportional to the number of variable nodes and does not change the overall space complexity. Thus, our improved algorithms maintain the same efficiency in both time and space as traditional BP, ensuring scalability and feasibility for practical implementation.

\section{Simulation Results}
\label{simulation results}

In all simulations, the criterion for successful decoding is that the error estimate satisfies $\hat{E} \in E\mathcal{S}$, rather than $\hat{E} = E$ as in classical information. All cases not meeting the above condition are referred to as \textit{logical errors} in this article. Herein, a \textit{detected error} refers to $\hat{E}E$ being anti-commutative with any one stabilizer $S \in \mathcal{S}$; an \textit{undetected error} implies $\hat{E}E$ commutes with all stabilizers, i.e., $\hat{E}E \in N(\mathcal{S}) \backslash \mathcal{S}$, where the result of error correction is a logical operator. Each data point on the experimental curves represents the mean of at least $10^4$ simulations.

Without further specification, our simulations employ a depolarizing channel, where each qubit has an equal probability $p/3$ of experiencing Pauli $X$, $Y$, or $Z$ errors, and a probability of $1-p$ of no error occurring. Additionally, for ease of comparison, we adopt the simple noise model which does not consider measurement errors or error propagation in quantum circuits.

\subsection{Breaking the Trapping Sets}
\label{Breaking the Trapping Sets}

\begin{figure*}[t]
  \centering
  \subfloat[Traditional LLR-BP\textsubscript{4}]{
    \includegraphics[width=0.95\linewidth]{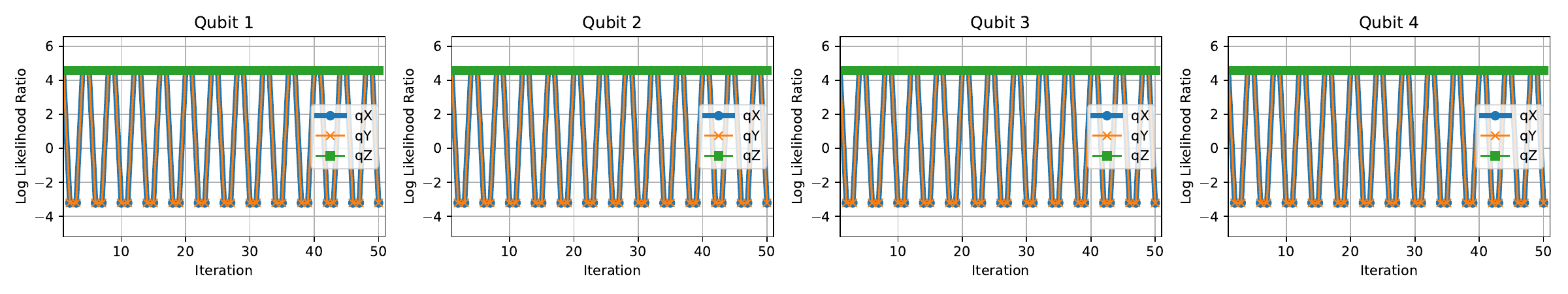}
  }
  
  \subfloat[AdaGrad-BP (parallel)]{
    \includegraphics[width=0.95\linewidth]{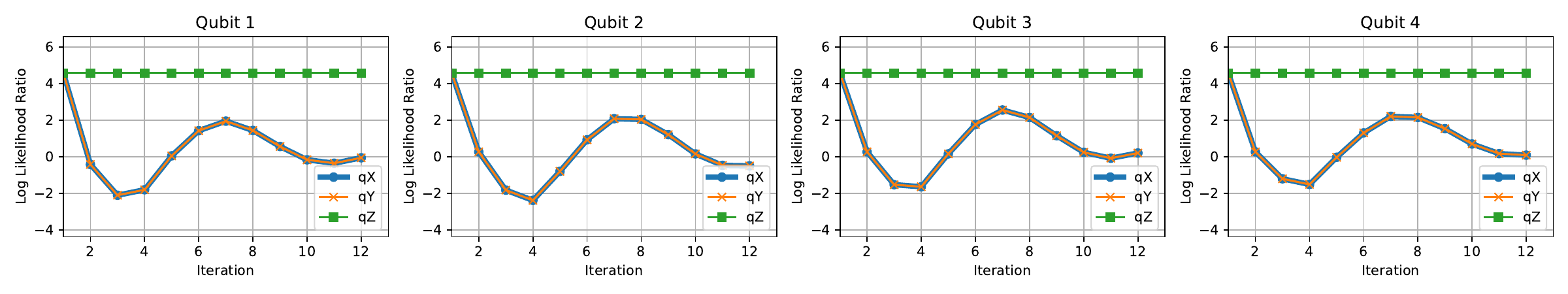}
  }

  \subfloat[AdaGrad-BP (serial)]{
    \includegraphics[width=0.95\linewidth]{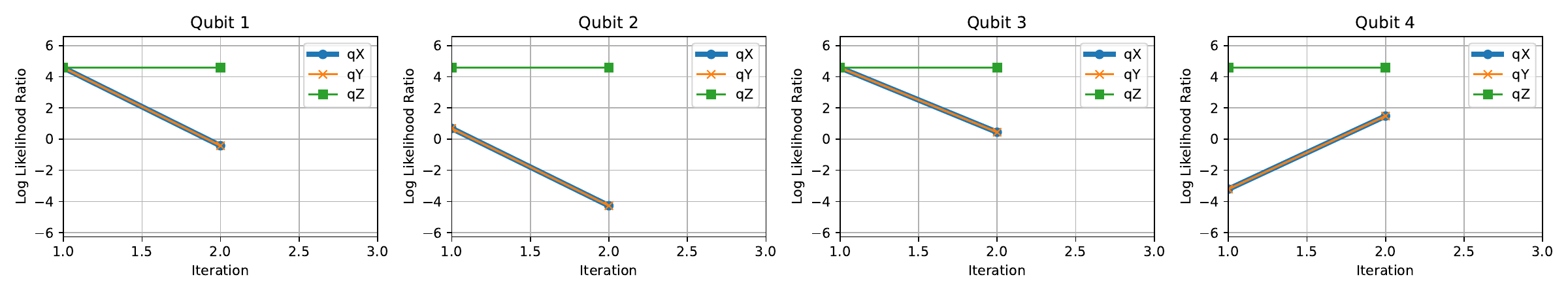}
  }

  \caption{Comparison of traditional LLR-BP\textsubscript{4} and AdaGrad-BP under parallel scheduling and serial scheduling on the (4,0) trapping set. The improved BP algorithm break the symmetry of message updates and converge within 12 and 2 iterations separately.}
  \label{fig:ts_decoding}
\end{figure*}

We first compare the performance of improved BP on the trapping set shown in Fig.~\ref{fig:ts}. In this example, qubit 1 and qubit 2 experience X errors, resulting in an error syndrome of $[1,0,1,0]$. Fig.~\ref{fig:ts_decoding} shows the comparison between the traditional LLR-BP\textsubscript{4} and AdaGrad-BP in Section \ref{Message update optimization} under parallel and serial scheduling, with the maximum iteration set to 50 and the \textit{a priori} probabilities set to 0.01. The hard decision result of the traditional BP oscillates continuously between all $X$ and all $I$, whereas improved BP algorithm show significantly reduced oscillation amplitude and converge within few iterations. 

A more detailed comparison of the different algorithms proposed in this article is presented in Table.~\ref{tb:ts_decoding}. As described in Section \ref{Message update optimization}, we set $\alpha$ of AdaGrad-BP to 5 and $\gamma$ of Momentum-BP to 0. For Momentum-BP, EWAInit-BP, and MBP, we tested 10 evenly distributed values for alpha within the range [0, 1] and selected the one that yielded the best performance. The same parameter selection strategy is also applied in Section \ref{Improved BP on Surface codes}. MBP (equals to AMBP in this simulation) \cite{kuo_2021_quantumBP9_MBP_gf4} and EWAInit-BP converges at the same fastest speed under serial scheduling but oscillates under parallel scheduling. BP-OTS \cite{Chytas_2024_quantumts3_BPOTS} shows little dependency on scheduling method but its convergence depends on the manually set oscillation period $ T $, which may vary in different error correction scenarios. Algorithms proposed in Section \ref{Message update optimization}, Momentum-BP and AdaGrad-BP, converge fast under both parallel scheduling and serial scheduling.

It is worth noting that EWAInit-BP fails to decode trapping sets under parallel scheduling. This is because the recursive message update formula for each bit remains symmetric, making it unable to resolve the issue within trapping sets. In practical error-correcting code decoding, this limitation may lead to non-convergence in very rare special cases. However, in most cases, messages originating from outside the trapping sets can assist in achieving convergence.

\begin{table}[t]
  \centering
  \caption{Iterations required by different algorithms on the (4,0) trapping set}
  \label{tb:ts_decoding}
  \small
  \begin{tabular}{lcc}
    \toprule
    Algorithm & Parallel Scheduling & Serial Scheduling \\
    \midrule
    LLR-BP\textsubscript{4} \cite{Lai_2021_quantumBP8_gf2} & \text{---} & \text{---} \\
    (A)MBP \cite{kuo_2021_quantumBP9_MBP_gf4} & \text{---} & $\alpha=0.5, iter=2$  \\
    BP-OTS \cite{Chytas_2024_quantumts3_BPOTS} & $iter=T+2$  & $iter=T+2$  \\ 
    EWAInit-BP & \text{---} & $\alpha=0.5, iter=2$  \\
    \textbf{AdaGrad-BP} & $iter=12$ & $iter=2$ \\
    \textbf{Momentum-BP} & $\alpha=0.5, iter=2$  & $\alpha=0.6, iter=2$  \\
    \bottomrule
  \end{tabular}
\end{table}

\subsection{Improved BP on Surface codes}
\label{Improved BP on Surface codes}

In this subsection, we compare the performance of proposed BP improvements with other improvements for various surface codes. For all BP algorithms, the maximum number of iterations $iter_{max}$ across all code lengths is set to 100 to ensure linear asymptotic complexity.

\subsubsection{Comparison of our BP optimization methods}

\begin{figure}[t]
  \centering
  \includegraphics[width=1\columnwidth]{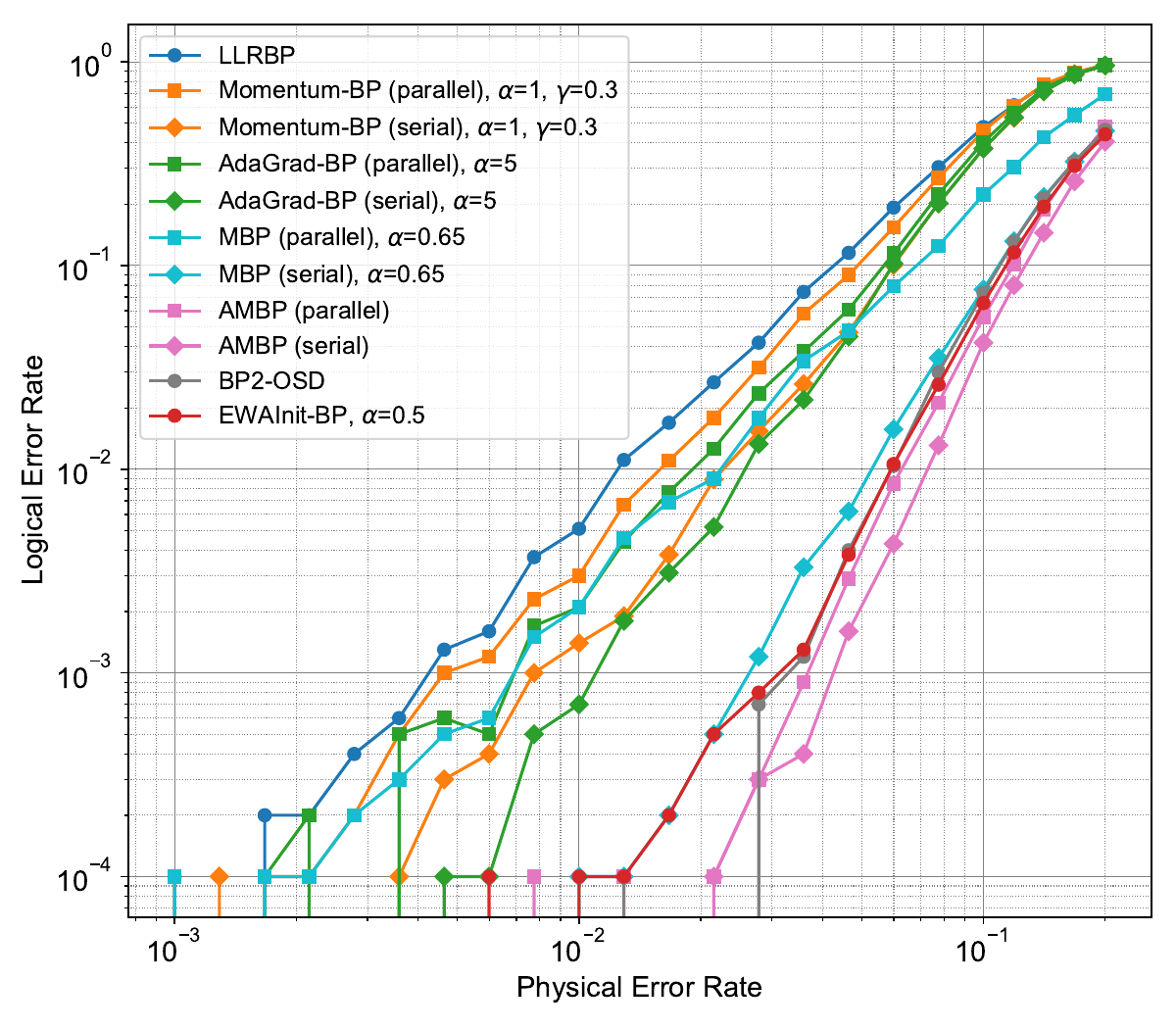}
  \caption{Overall Performance of proposed BP algorithms with existing algorithms on a $[[85, 1, 7]]$ planar surface code. 
  Momentum-BP, AdaGrad-BP, MBP and AMBP are evaluated under both parallel and serial scheduling. EWAInit-BP performs between serial MBP and serial AMBP and is almost identical to BP\textsubscript{2}-OSD-0 in this case.}
  \label{fig:paper_compare}
\end{figure}

A comparison of all methods proposed in this article with some existing high-accuarcy BP improvements on a $L=7$ planar surface code is illustrated in Fig.~\ref{fig:paper_compare}, using the same parameter selection strategy as described in Section \ref{Breaking the Trapping Sets}. The two methods for optimizing message update rules, Momentum-BP and AdaGrad-BP, exhibit improvements over traditional LLR-BP\textsubscript{4} at physical error rates lower than 10\%, providing evidence for the validity of the gradient perspective of BP. The prior-probability-based method, EWAInit-BP, demonstrates improvements ranging from 1 to 3 orders of magnitude over traditional BP and effectively eliminates the error floor. Compared to other algorithms, EWAInit-BP achieves performance between serial MBP and serial AMBP while using hardware-friendly parallel scheduling, and in this case, performs nearly identically to BP\textsubscript{2}-OSD-0 \cite{roffe_2020_quantumBP7_gf2}, despite the latter having a time complexity of $ O(n^3) $. 

The enhancements observed with Momentum-BP and AdaGrad-BP support the validity of the gradient perspective, providing a foundation for further research in this direction. Next, we will provide detailed comparisons between our best enhancement, EWAInit-BP and algorithms with similar time overheads (excluding post-processing and external loops). Although AMBP\cite{kuo_2021_quantumBP9_MBP_gf4} has an asymptotic complexity of $O(N)$, its actual runtime overhead is high, as detailed in Appendix \ref{AMBP}, which also includes the results for Adaptive EWA-BP.

\subsubsection{Toric code}

We construct the $[[2n^2, 1, n]]$ toric code based on its properties as a hypergraph product of an $[n, 1]$ repetition code. The toric code features periodic boundaries, hence all its stabilizers have a weight of four, and any error syndromes on it exhibit translational symmetry \cite{Andreasson_2019_DRL}, which is particularly conducive to BP and neural network decoding. We compare the performance of our EWAInit-BP and other existing methods under both parallel and serial scheduling as shown in Fig.~\ref{fig:toric_flood} and Fig.~\ref{fig:toric_layer}, respectively. Note that in all simulations in this subsection, the parameter choices for different code lengths are consistent.

\begin{figure}[t]
  \centering
  \subfloat[Logical error rates]{
    \includegraphics[width=\columnwidth]{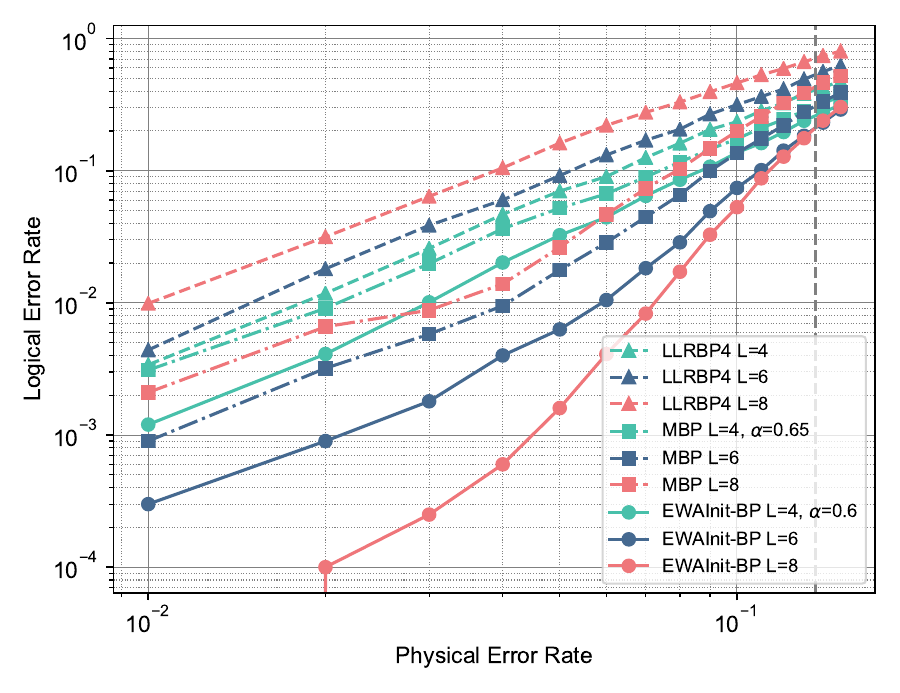}
    \label{fig:toric_flood_LER}
  }

  \subfloat[Average iterations]{
    \includegraphics[width=\columnwidth]{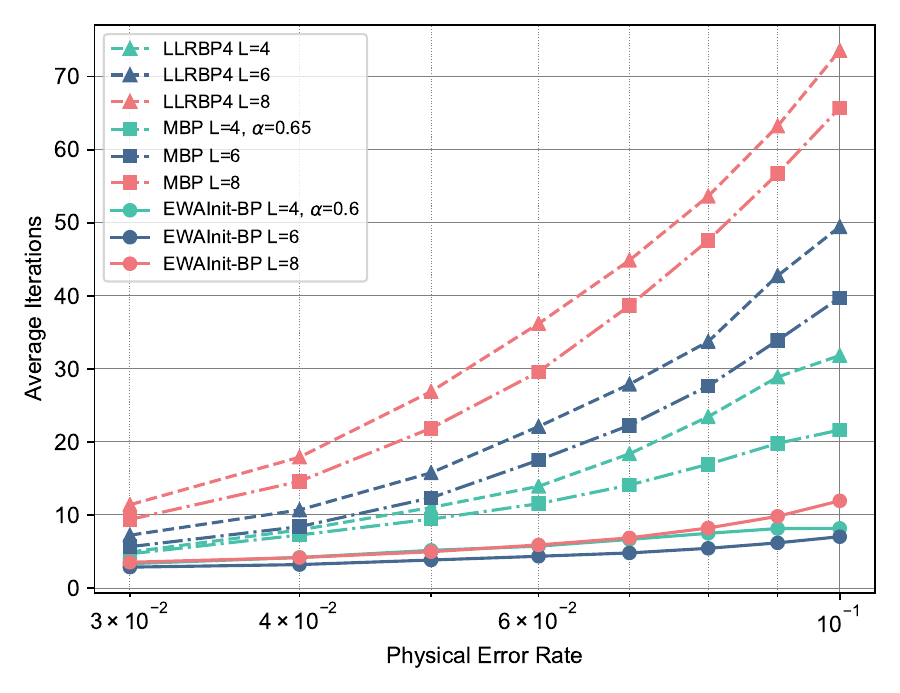}
    \label{fig:toric_flood_iter}
  }

  \caption{Performance of LLRBP4, MBP, and EWAInit-BP under parallel scheduling for toric codes with L = \{4, 6, 8\}. Only EWAInit-BP find a reduction in the logical error rate as the lattice size increases, with the intersection of error rates occurring at 13.6\%.}
  \label{fig:toric_flood}
\end{figure}

\begin{figure}[t]
  \centering
  \subfloat[Logical error rates]{
    \includegraphics[width=\columnwidth]{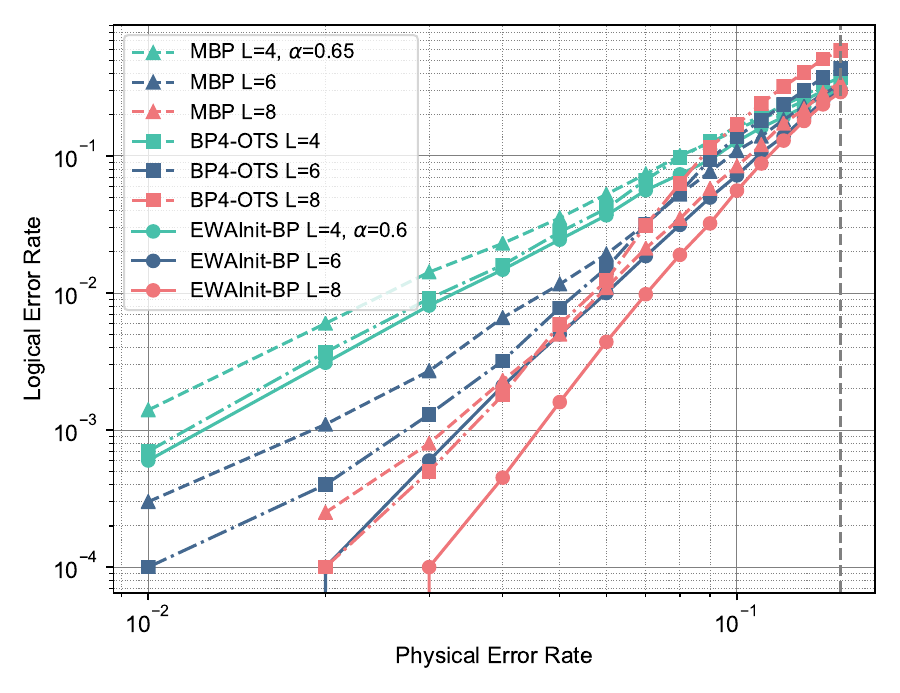}
    \label{fig:toric_layer_LER}
  }

  \subfloat[Average iterations]{
    \includegraphics[width=\columnwidth]{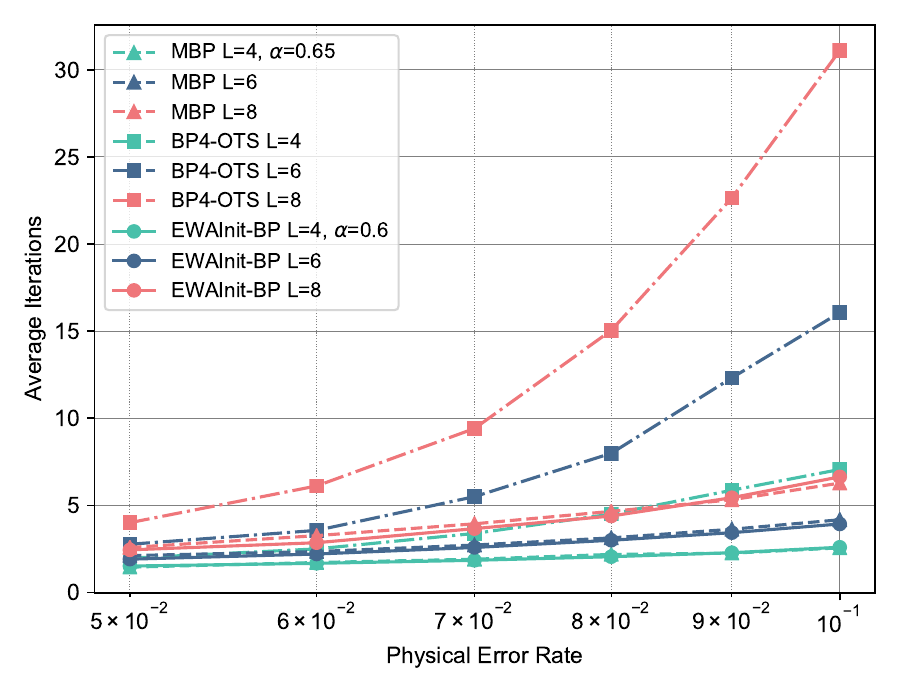}
    \label{fig:toric_layer_iter}
  }

  \caption{Performance of MBP, BP-OTS, and EWAInit-BP under serial scheduling for toric codes with L = \{4, 6, 8\}. We choose $ T=9 $ and $ C=20 $ for BP-OTS. EWAInit-BP outperforms other BP improvements, with the intersection of logical error rates at 15\%.}
  \label{fig:toric_layer}
\end{figure}

For the parallel schedule\footnote{Our implementation of BP-OTS over GF(4) under parallel scheduling did not achieve the performance reported in the original paper, necessitating further investigation into this discrepancy.}, compared to the classic BP over GF(4) with LLR message and its enhancement, MBP, our EWAInit-BP is the only method that achieves a decrease in logical error rates as the lattice size increases, with significant improvements in logical error rates at the same lattice size. The intersection of the error rate curves for these lattice sizes is approximately 13.6\%. Additionally, the average number of iterations required for EWAInit-BP is significantly lower than the other two methods, converging in fewer than 10 iterations mostly. Parallel scheduling of BP algorithms holds the potential to achieve constant complexity through hardware-based parallel computation, while EWAInit-BP ensures both the error correction capabilities and speed in decoding.

For the serial schedule, all three decoding algorithms demonstrated error correction capabilities. Among those, BP-OTS performs better than MBP at low physical error rates, but its performance is surpassed by MBP when the physical error rate exceeds 7\%. The performance of EWAInit-BP, on the other hand, consistently outperforms these two methods across all tested conditions, with the logical error rate at $L=6$ being comparable to that of the other two algorithms at $L=8$, and the intersection of the logical error rate curves is approximately 15\%\footnote{We cannot yet consider the intersection of curves for short codes as the obtained threshold, as the convergence capability of EWAInit-BP degrades to some extent on larger codes. This will be discussed in Appendix \ref{AMBP}.}. In terms of iteration count, both EWAInit-BP and MBP converge in remarkably few iterations (often fewer than 5).

\subsubsection{Planar surface code}

The boundaries of the planar surface code are no longer contiguous, resulting in the emergence of some stabilizers with weight three and qubits that only connect two stabilizers at the boundaries. This configuration leads to a slight reduction in the error correction capabilities of the planar surface code compared to the Toric code. We compare the decoding logical error rates of LLR-BP\textsubscript{4}, MBP, and EWAInit-BP under both parallel and serial scheduling, as illustrated in Fig.~\ref{fig:surface_code}.

\begin{figure}[t]
  \centering
  \subfloat[LER under a parallel schedule]{
    \includegraphics[width=\columnwidth]{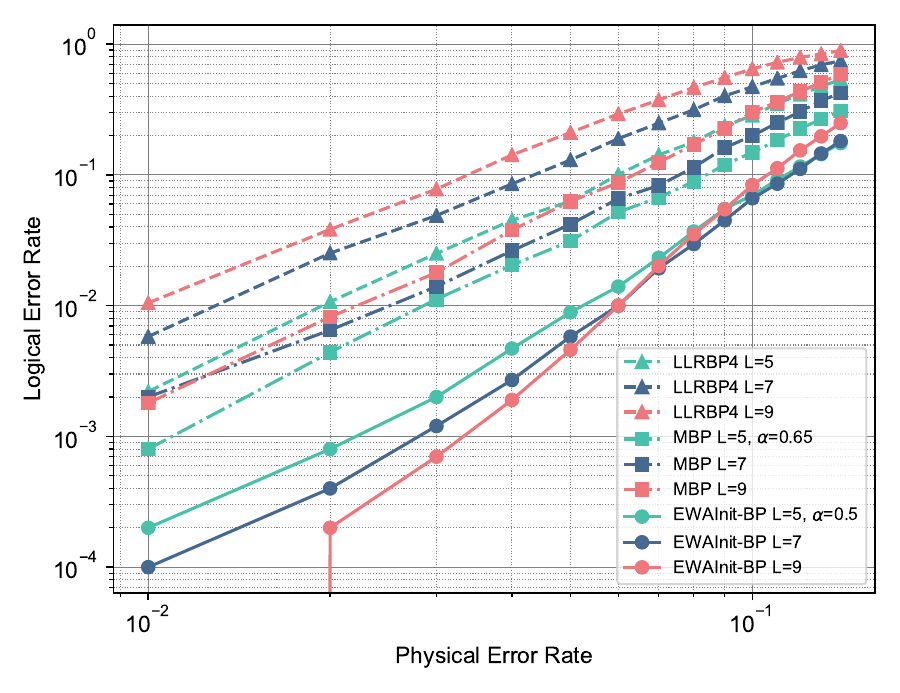}
    \label{fig:surface_flood}
  }

  \subfloat[LER under a serial schedule]{
    \includegraphics[width=\columnwidth]{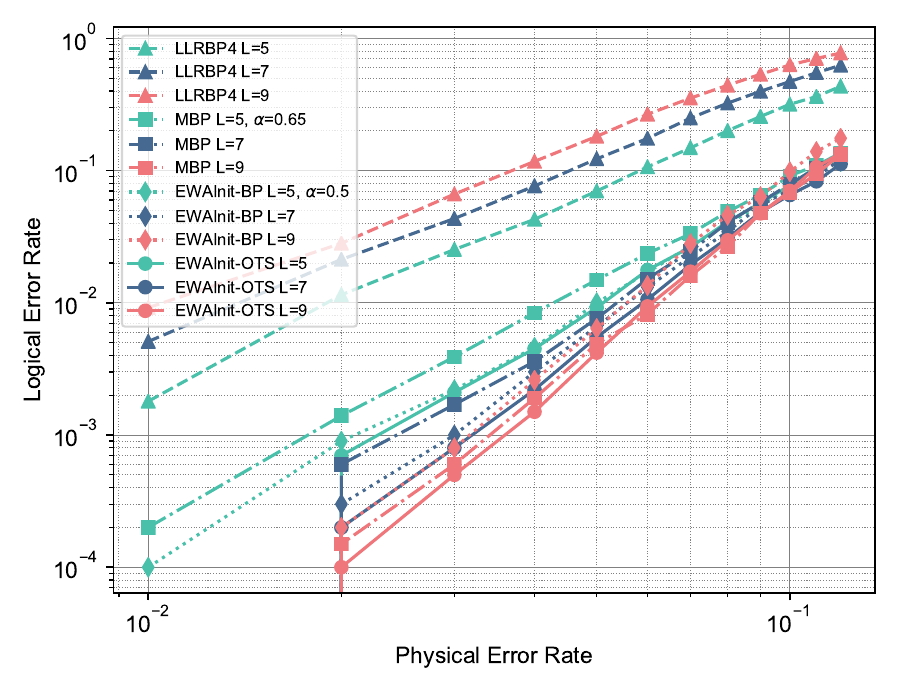}
    \label{fig:surface_layer}
  }

  \caption{Performance of LLRBP4, MBP, and EWAInit-BP (EWAInit-OTS) for planar surface codes with L = 5, 7, 9. The intersection of logical error rates for EWAInit-BP under serial scheduling is approximately 10\%.}
  \label{fig:surface_code}
\end{figure}

Under parallel scheduling, the decoding logical error rates of LLRBP4 and MBP4 increase with increasing code length, whereas EWAInit-BP achieves a reduction in logical error rates at low physical error rates. Unfortunately, at higher error rates, decoding at higher code lengths becomes more challenging, preventing EWAInit-BP from achieving a reliable threshold on the planar surface code.

Under serial scheduling, EWAInit-BP performs better than MBP at $L=5$ and $L=7$; however, its performance at $L=9$ is not satisfying. Despite the dynamic updating of \textit{a priori} probabilities, the method still partly relies on the symmetric initial input priors. One solution is to incorporate the OTS method as a "post-processing" step for EWAInit, resetting the \textit{a priori} probabilities of some nodes every few iterations, with the hyperparameters $T=5$ and $C=20$. This EWAInit-OTS method, while maintaining a time complexity of $O(n)$, shows improvements in convergence rate and logical error rates on the surface code compared to EWAInit-BP and other methods.

The patterns of average iterations of different algorithms for the planar surface code and the XZZX surface code are similar to those observed in the toric code and exhibit the same trends as the logical error rates. Therefore, we do not repeat the iteration figures here.

\subsubsection{XZZX Surface code}

For biased noise, the XZZX surface code exhibits superior error correction capabilities compared to the traditional planar surface code, achieving the theoretical maximum threshold of 50\% across three types of pure Pauli error channels. 

A comparison of decoding performance on the XZZX surface code is illustrated in Fig.~\ref{fig:xzzx_surface_code}. Our EWAInit combined with OTS method outperforms other algorithms under serial scheduling, with the LER curve intersection at approximately 12.4\%. Unfortunately, for odd lattice sizes, the logical error rate for higher code lengths is higher than for lower code lengths at low physical error rates, while it demonstrates error-correcting capabilities at higher physical error rates.

\begin{figure}[t]
  \centering
  \includegraphics[width=\linewidth]{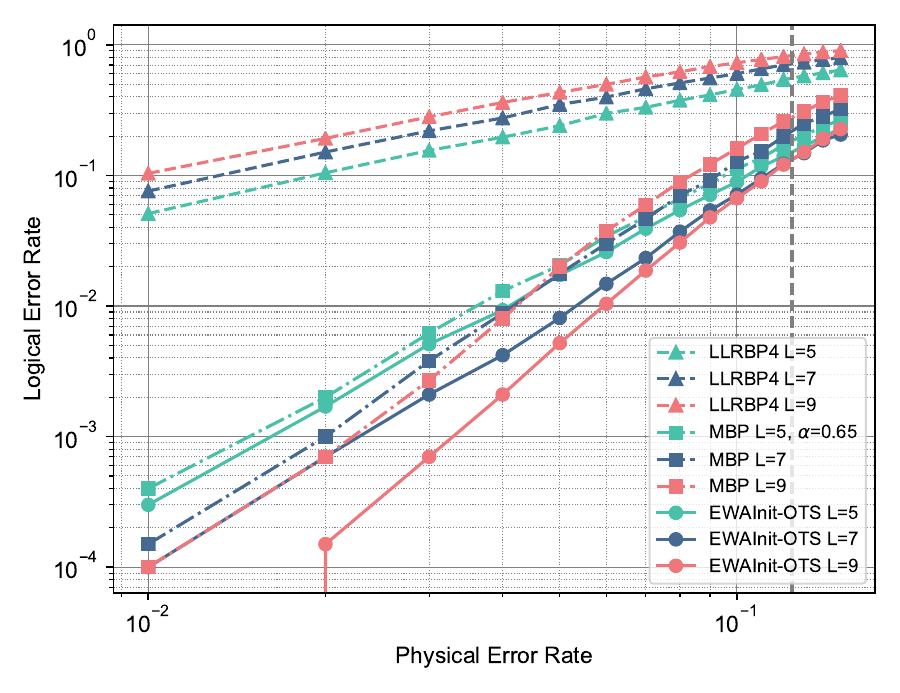}
  \caption{Performance of LLRBP4, MBP, and EWAInit-BP (EWAInit-OTS) for XZZX surface codes with L = 5, 7, 9 under serial scheduling. The intersection of logical error rates for EWAInit-BP is approximately 12.4\%.}
  \label{fig:xzzx_surface_code}
\end{figure}

To demonstrate the bias noise correction capability of the XZZX surface code, we chose $ p_z = \frac{2}{3} $ and $ p_x = p_y = \frac{1}{6} $, implying that Pauli Z errors are more likely to occur, which aligns with the conditions found in actual hardware. Fig.~\ref{fig:xzzx_bias} displays the decoding performance under biased noises under parallel scheduling.
\begin{figure}[t]
  \centering
  \includegraphics[width=\linewidth]{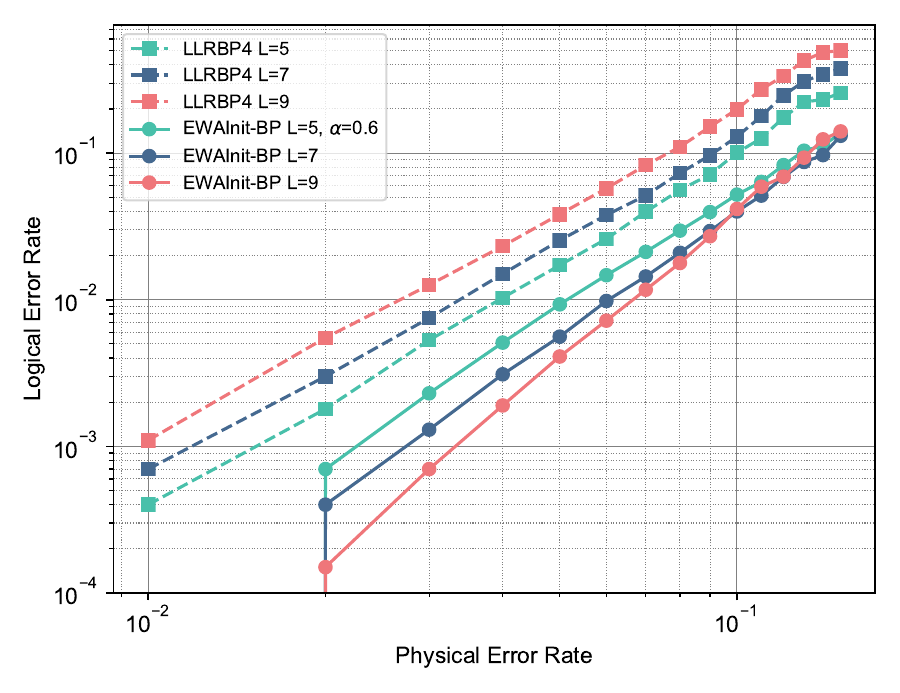}
  \caption{Performance of LLRBP4 and EWAInit-BP for the XZZX surface code with L = 5, 7, 9 under parallel scheduling, evaluated under conditions where Pauli Z errors are more prevalent in the channel. The intersection of LERs for EWAInit-BP is approximately 10\%.}
  \label{fig:xzzx_bias}
\end{figure}
It is evident that when biased noise is present, the performance of BP decoding using the XZZX surface code with parallel scheduling shows significant improvement compared to that on a depolarizing channel. Besides, EWAInit-BP offers an improvement of approximately one order of magnitude and demonstrates the error correction capability under parallel scheduling, with the intersection of logical error rates around 10\%.

\subsection{Generality on other quantum LDPC codes}
\label{sec:generality_qldpc}

Although the trapping set analysis presented in Section~\ref{Breaking the Trapping Sets} is difficult to directly extend to arbitrary quantum LDPC codes due to the structural differences among codes, the performance of our EWAInit-BP algorithm is not limited to 2D topological constructions. In this subsection, we provide empirical evidence that EWAInit-BP generalizes well to other classes of quantum LDPC codes. Two representative examples are considered: the Bivariate Bicycle codes proposed by IBM~\cite{bravyi_2024_IBM}, and 3D Chamon codes which is an instance of 3D XYZ product of three repetition codes~\cite{leverrier2022xyz}.

Fig.~\ref{fig:bicycle_chamon} presents a comparison of logical error rates between BP\textsubscript{2}-OSD-0 and EWAInit-BP under parallel scheduling on these two code families under depolarizing noise and code capacity noise model. For the Bivariate Bicycle Codes, EWAInit-BP almost consistently outperforms BP\textsubscript{2}-OSD-0 across all tested code distances, with its advantage more pronounced at lower physical error rates.

Similarly, for the 3D Chamon codes, EWAInit-BP demonstrates improved decoding performance for all tested distances. While both algorithms show typical threshold-like behavior, EWAInit-BP achieves lower logical error rates under comparable physical error rates. These results provide preliminary evidence that the exponential averaging mechanism is broadly beneficial for decoding quantum LDPC codes.

\begin{figure}[t]
  \centering
  \subfloat[Bivariate Bicycle Codes]{
    \includegraphics[width=0.95\columnwidth]{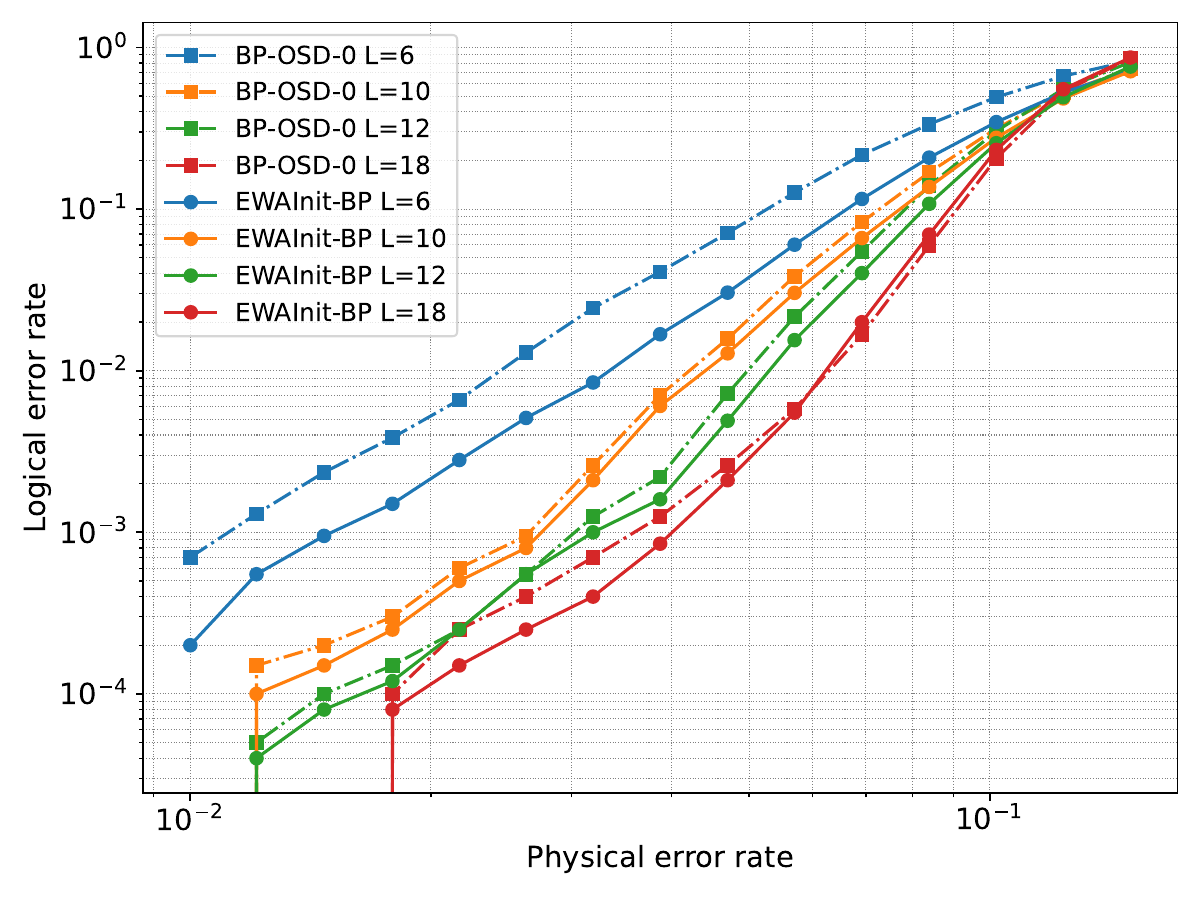}
  }

  \subfloat[3D Chamon Codes]{
    \includegraphics[width=0.95\columnwidth]{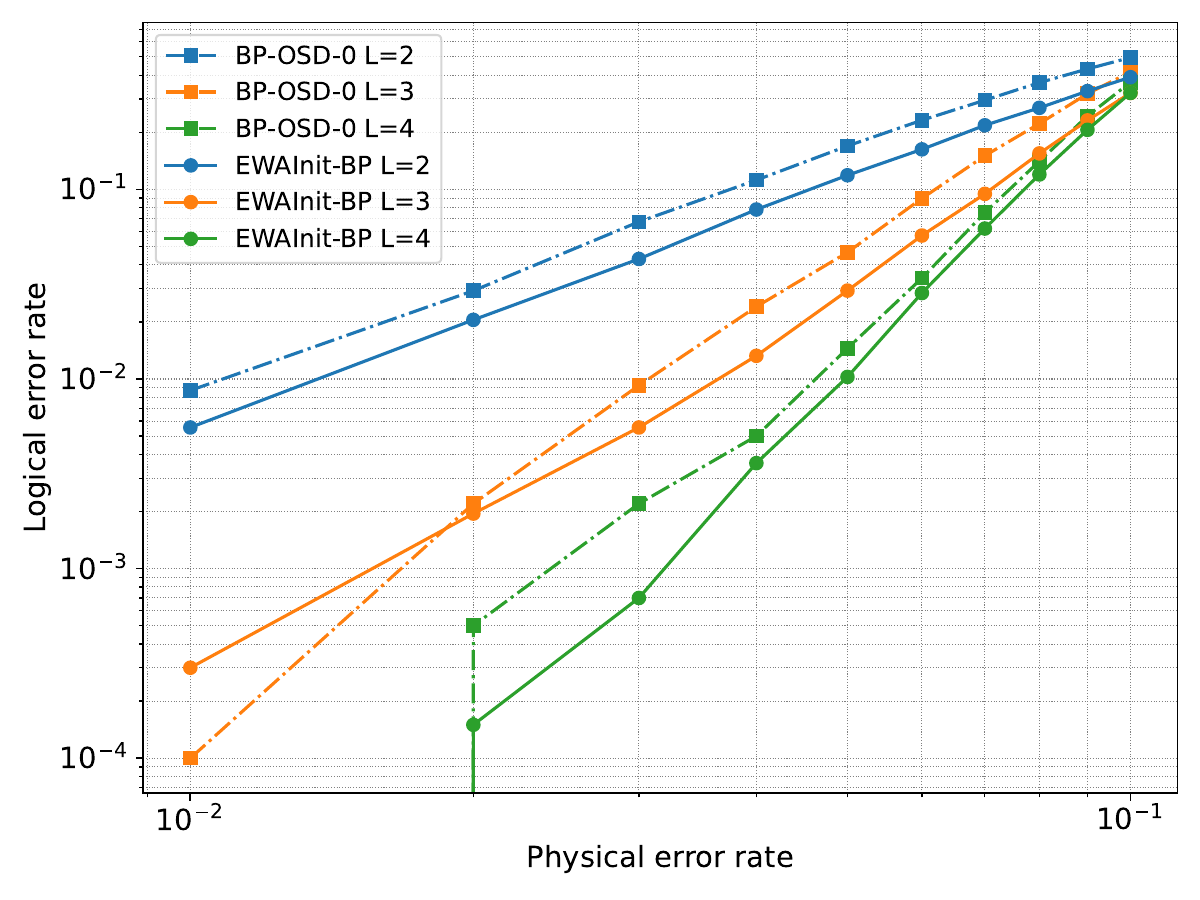}
  }

  \caption{Comparison of BP\textsubscript{2}-OSD-0 and EWAInit-BP on two non-surface-code quantum LDPC families: (a) Bivariate Bicycle Codes and (b) 3D Chamon Codes. EWAInit-BP almost outperforms BP-OSD-0 across all distances and physical error rates.}
  \label{fig:bicycle_chamon}
\end{figure}

\section{Conclusion}
\label{summary}

In this article, we proposed three improved BP algorithms: Momentum-BP, AdaGrad-BP and EWAInit-BP. 

Since the \textit{a posteriori} probability update of BP can be interpreted as the gradient optimization of an energy function, we adopted optimization strategies from gradient descent in machine learning, resulting in Momentum-BP and AdaGrad-BP. These mechanisms aim to smooth the update process and reduce the amplitude of oscillations. Simulations on trapping sets intuitively demonstrated the effectiveness of these algorithms, supporting the feasibility of leveraging gradient optimization principles for BP decoding.

From the theoretical derivation of the message update rules of BP algorithm, we recognized the importance of \textit{a priori} probabilities and transformed the message update smoothing into initial values adapting, proposing EWAInit-BP. This algorithm realizes an exponentially weighted average of past horizontal message updates and outperforms other algorithms without OSD post-processing on various surface codes. Specifically, EWAInit-BP shows high decoding accuracy even under parallel scheduling, with logical error rates decreasing as code length increases on short blocks. Notably, all our proposed methods do not introduce any additional factors to the time complexity of BP. Thus, EWAInit-BP can be implemented with $O(1)$ complexity in fully parallelism while maintaining decoding accuracy.

Furthermore, we also evaluated the performance of EWAInit-BP on two representative quantum LDPC codes—Bivariate Bicycle Codes and the 3D Chamon codes. Our algorithm demonstrates competitive decoding performance on both families under parallel scheduling, showing potential adaptability to high-dimensional or irregular topologies. The results suggest that the principle behind message adaptation strategies may extend beyond the specific context of surface codes, indicating broader theoretical significance. Future work will explore how to adapt our approach to other quantum LDPC codes with higher degeneracy or more complex structures, especially focusing on maintaining decoding performance at larger code distances.

However, applying exponential weighting to messages is still insufficient to ensure BP convergence for all code lengths. It may be necessary to draw inspiration from methods like neural BP to enhance the representational power of the algorithm. Besides, our simulations only considered Pauli errors on qubits and did not account for measurement errors, error propagation in circuits, etc. Future work needs to adapt BP to the circuit-level noise model, potentially requiring specialized optimizations.

\section{Appendix}
\label{appendix}

\subsection{BP as gradient optimization}
\label{app:gradient}

Each iteration of BP can be interpreted as one step of gradient descent with a simplified step size on the following energy function \cite{lucas_1998_BPGradient, kuo_2021_quantumBP9_MBP_gf4}
\begin{align}
    J(\boldsymbol{Q}) = & -\sum_{i=1}^m 2 \tanh^{-1} \nonumber \\
    & \left( (-1)^{s_i} \prod_{v_j \in \mathcal{N}(c_i)} \tanh \left( \frac{\lambda_{H_{ij}}(Q_{v_j})}{2} \right) \right),
\end{align}
where $\boldsymbol{Q}$ is the \textit{a posteriori} probabilities of all the variable nodes in this ieration, $m$ is the number of check nodes. The term inside the parentheses measures the discrepancy between the result of the current iteration and the actual error syndrome of stabilizer $ S_i $. The smaller the energy function, the closer the BP iteration result is to the correct solution. 

The variables for gradient updates are the \textit{a posteriori} estimates of each type of Pauli error of each variable node. Specifically, for the $ W \in \{X, Y, Z\} $ type Pauli error on the $ n $-th variable node, the \textit{a posteriori} probability update is
\begin{align}
    Q_{v_j, W}^{(t)} = Q_{v_j, W}^{(t-1)} - \alpha \frac{\partial J}{\partial Q_{v_j, W}^{(t-1)}},
\end{align}
where the partial derivative is
\begin{align}
    \frac{\partial J}{\partial Q_{v_j,W}} = & -\sum_{\substack{c_i \in \mathcal{M}(v_j) \\ \left\langle W, H_{ij} \right\rangle = 1}} \widetilde{m}_{c_i \rightarrow v_j} \nonumber \\
    & * \frac{g_{ij}(Q) \  e^{-Q_{v_j, W}}}{e^{-Q_{v_j, X}} + e^{-Q_{v_j, Y}} + e^{-Q_{v_j, Z}} - e^{-Q_{v_j, H_{ij}}}} \label{partial},
\end{align}
where $g_{ij}$ is a term that is always greater than 0 (see \cite{kuo_2021_quantumBP9_MBP_gf4}), and $ \widetilde{m}_{c_i \rightarrow v_j} $ is the \textit{a posteriori} probability summation term, approximated using the $ \tanh $ function
\begin{align}
    \widetilde{m}_{c_i \rightarrow v_j} = (-1)^{s_i} \prod_{v_{j'} \in \mathcal{N}(c_i) \backslash v_j} \tanh \left( \frac{\lambda_{H_{ij'}}(Q_{v_{j'}})}{2} \right).
\end{align}

If we view the last term on the RHS of \eqref{partial} as a variable coefficient, then the posterior update in BP can be seen as a gradient descent step that omits this coefficient term.

\subsection{Iterative expansion of EWA-based BP algorithms}
\label{app:expansion}

Here we present the expansion of Momentum-BP and EWAInit-BP proposed in this article to clarify their properties. To streamline calculations and maintain consistency, we set $\gamma$ in Momentum-BP to 0.

\subsubsection{Momentum-BP}
For Momentum-BP, the \textit{a posterori} probabilities are updated in the first iteration as
\begin{align}
    Q_{v_j, W}^{(1)} &= \alpha (\Pi_{v_j, W}^{(0)} + \sum_{\substack{c_i \in \mathcal{M}(v_j) \\ \left\langle W, H_{ij}\right\rangle=1}} m_{c_i \rightarrow v_j}^{(1)}) + (1-\alpha) \Pi_{v_j, W}^{(0)} \\
    &= \Pi_{v_j, W}^{(0)} + \alpha \sum_{\substack{c_i \in \mathcal{M}(v_j) \\ \left\langle W, H_{ij}\right\rangle=1}} m_{c_i \rightarrow v_j}^{(1)},
\end{align}
where the first term on the RHS represents the calculation in the current iteration. The second term represents the calculation from the pervious iteration, which in the first iteration is simply the \textit{a priori} probability, say $Q_{v_j, W}^{(0)} = \Pi_{v_j, W}^{(0)}$. 

Now we can expand the \textit{a posterori} update iteratively,
\begin{align}
    Q_{v_j, W}^{(t)} &= \alpha (\Pi_{v_j, W}^{(0)} + \sum_{\substack{c_i \in \mathcal{M}(v_j) \\ \left\langle W, H_{ij}\right\rangle=1}} m_{c_i \rightarrow v_j}^{(t)}) + (1-\alpha) Q_{v_j, W}^{(t-1)} \\
    &= \alpha (\Pi_{v_j, W}^{(0)} + \sum_{\substack{c_i \in \mathcal{M}(v_j) \\ \left\langle W, H_{ij}\right\rangle=1}} m_{c_i \rightarrow v_j}^{(t)}) + (1-\alpha) * \nonumber \\
    & \quad ( \alpha (\Pi_{v_j, W}^{(0)} + \sum_{\substack{c_i \in \mathcal{M}(v_j) \\ \left\langle W, H_{ij}\right\rangle=1}} m_{c_i \rightarrow v_j}^{(t-1)}) + (1-\alpha) Q_{v_j, W}^{(t-2)}) \\
    &= \alpha \sum_{k=0}^{t-1} (1 - \alpha)^k \left( \Pi_{v_j, W}^{(0)} + \sum_{\substack{c_i \in \mathcal{M}(v_j) \\ \langle W, H_{ij} \rangle = 1}} m_{c_i \rightarrow v_j}^{(t-k)} \right) \nonumber \\ 
    & \quad + (1 - \alpha)^t Q_{v_j, W}^{(0)}. \label{Momentum-general}
\end{align}

Using the initial condition $Q_{v_j, W}^{(0)} = \Pi_{v_j, W}^{(0)}$, we can substitute $Q_{v_j, W}^{(0)}$ in \eqref{Momentum-general} to obtain
\begin{align}
    Q_{v_j, W}^{(t)} &= \Pi_{v_j, W}^{(0)} \left( \alpha \sum_{k=0}^{t-1} (1 - \alpha)^k + (1 - \alpha)^t \right) \nonumber \\  
    & \quad + \alpha \sum_{k=0}^{t-1} (1 - \alpha)^k \sum_{\substack{c_i \in \mathcal{M}(v_j) \\ \langle W, H_{ij} \rangle = 1}} m_{c_i \rightarrow v_j}^{(t-k)}.
\end{align}
Now, let’s examine the coefficient of $\Pi_{v_j, W}^{(0)}$: 
\begin{align}
    \alpha \sum_{k=0}^{t-1} (1 - \alpha)^k + (1 - \alpha)^t. \label{Momentum-coefficient}
\end{align}
This is a finite geometric series, and we can simplify it using the formula
\begin{align}
    \sum_{k=0}^{n} r^k = \frac{1 - r^{n+1}}{1 - r}.
\end{align}
Substituting back to \eqref{Momentum-coefficient}, we get
\begin{align}
    \alpha \sum_{k=0}^{t-1} (1 - \alpha)^k &+ (1 - \alpha)^t = \\ 
    & \alpha \cdot \frac{1 - (1 - \alpha)^t}{\alpha} + (1 - \alpha)^t = 1.
\end{align}

Thus, we have shown that the coefficient of $\Pi_{v_j, W}^{(0)}$ is indeed 1. The complete general term formula is therefore
\begin{align}
    Q_{v_j, W}^{(t)} = \Pi_{v_j, W}^{(0)} + \sum_{k=0}^{t} \alpha (1 - \alpha)^k \sum_{\substack{c_i \in \mathcal{M}(v_j) \\ \langle W, H_{ij} \rangle = 1}} m_{c_i \rightarrow v_j}^{(t-k)}.
\end{align}

\subsubsection{EWAInit-BP}

Similarly, for EWAInit-BP, we have the initial condition
\begin{align}
    \Pi_{v_j, W}^{(1)} &= \Pi_{v_j, W}^{(0)}, \\
    Q_{v_j, W}^{(1)} &= \Pi_{v_j, W}^{(1)} + \sum_{\substack{c_i \in \mathcal{M}(v_j) \\ \langle W, H_{ij} \rangle = 1}} m_{c_i \rightarrow v_j}^{(1)},
\end{align}
and the general term form
\begin{align}
    \Pi_{v_j, W}^{(t)} &= \alpha \Pi_{v_j, W}^{(0)} + (1-\alpha) Q_{v_j, W}^{(t-1)} \\
    &= \Pi_{v_j, W}^{(0)} + \sum_{k=1}^{t} (1 - \alpha)^k \sum_{\substack{c_i \in \mathcal{M}(v_j) \\ \langle W, H_{ij} \rangle = 1}} m_{c_i \rightarrow v_j}^{(t-k)}, \\
    Q_{v_j, W}^{(t)} &= \Pi_{v_j, W}^{(t)} + \sum_{\substack{c_i \in \mathcal{M}(v_j) \\ \langle W, H_{ij} \rangle = 1}} m_{c_i \rightarrow v_j}^{(t)} \\
    &= \Pi_{v_j, W}^{(0)} + \sum_{k=0}^{t} (1 - \alpha)^k \sum_{\substack{c_i \in \mathcal{M}(v_j) \\ \langle W, H_{ij} \rangle = 1}} m_{c_i \rightarrow v_j}^{(t-k)},
\end{align}
where the coefficient of $\Pi_{v_j, W}^{(0)}$ is also 1, derived by the same approach.

Although the approaches for improvement are different, both Momentum-BP and EWAInit-BP attempt to enhance the update of the \textit{a posteriori} probability $Q_{v_j, W}^{(t)}$ by applying a weighted combination of past messages. The distinction lies in the fact that the recurrence relation derived from performing EWA on the prior probability is more concise and intuitive, with only a $ (1 - \alpha) $ factor applied to past messages. Interestingly, EWAInit-BP, with its simpler recurrence relation, also demonstrates better performance in simulations.

\subsection{Compared to Adaptive MBP}
\label{AMBP}

Adaptive MBP (AMBP)\cite{kuo_2021_quantumBP9_MBP_gf4} defines a parameter range (for surface codes, $\alpha \in [1, 0.5]$) and iterates through the possible values from large to small until MBP converges with a specific parameter $\alpha^*$, at which point it returns the result. Sorting the parameters from large to small prevents overly aggressive message updates that might cause logical errors, while ensuring convergence in most cases. In this appendix, we uniformly sampled 11 values within the range $[1, 0.5]$, making the worst-case time overhead for AMBP approximately 10 times that of MBP. We can utilize a similar adaptive parameter selection strategy for EWAInit-BP using a parameter range of $\alpha \in [1, 0]$, where the parameter represents the strength of the EWA. We refer to the resulting algorithm as AEWA-BP (Adaptive EWAInit-BP).

We compare EWAInit-BP, AEWA-BP with parallel scheduling, and AMBP with both parallel and serial scheduling on planar surface codes of sizes $[[13, 1, 3]]$, $[[85, 1, 7]]$, and $[[221, 1, 11]]$, as shown in Fig.~\ref{fig:AMBP_LER}. EWAInit-BP and parallel AMBP no longer exhibit the property of decreasing logical error rates with increasing code length, indicating that large surface codes remain a challenge for parallel BP. However, AEWA-BP with parallel scheduling demonstrates accuracy comparable to AMBP with serial scheduling. A major issue with AEWA-BP is that its parameter range cannot cover all possible magnitudes of message updates, meaning it theoretically cannot guarantee convergence for all errors.

\begin{figure}[t]
  \centering
  \includegraphics[width=0.9\linewidth]{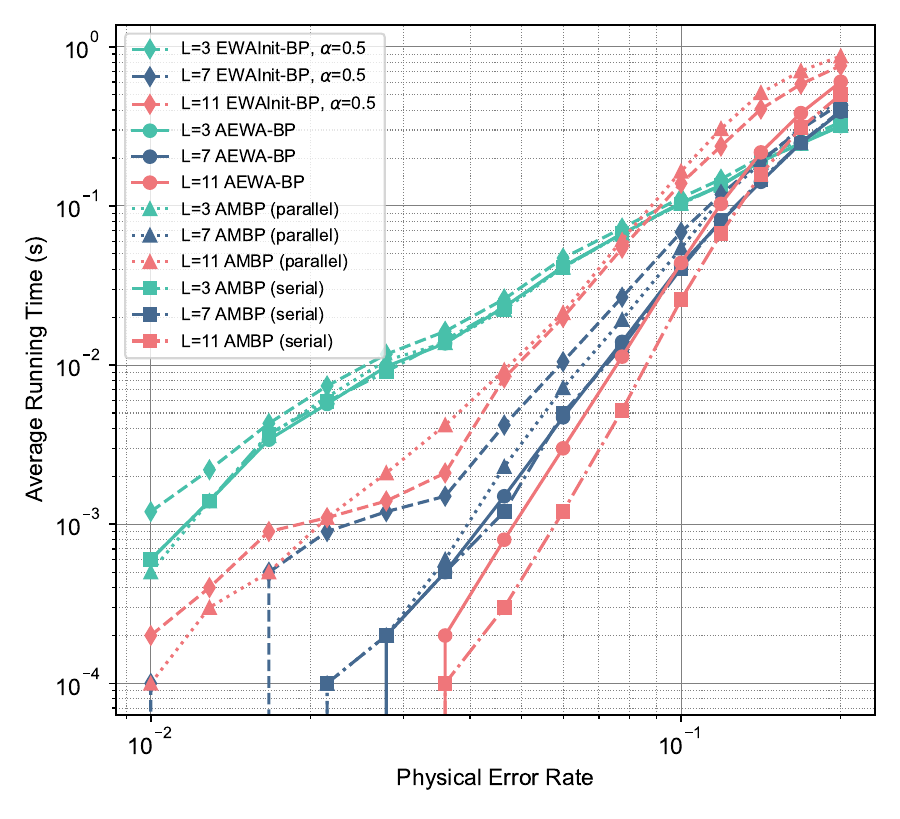}
  \caption{Performance of EWAInit-BP, AEWA-BP, parallel AMBP, and serial AMBP on the planar surface code with L = 3, 7, 11. EWAInit-BP and parallel AMBP achieve similar accuracy, but both exhibit higher error rates at \( L = 11 \) compared to \( L = 7 \). AEWA-BP achieves accuracy comparable to serial AMBP.}
  \label{fig:AMBP_LER}
\end{figure}

Next, we compared the runtime of the aforementioned algorithms to verify the high overhead of adaptive methods discussed in the main text, as shown in Fig.~\ref{fig:AMBP_time}. The runtime of parallel AMBP is very similar to that of AEWA-BP and is therefore omitted from the figure. As shown, EWAInit-BP exhibits the same scaling properties as traditional BP: its runtime grows linearly with code length and increases with physical error rates due to the rising number of iterations (up to a maximum of 100). 

\begin{figure}[t]
  \centering
  \includegraphics[width=0.9\linewidth]{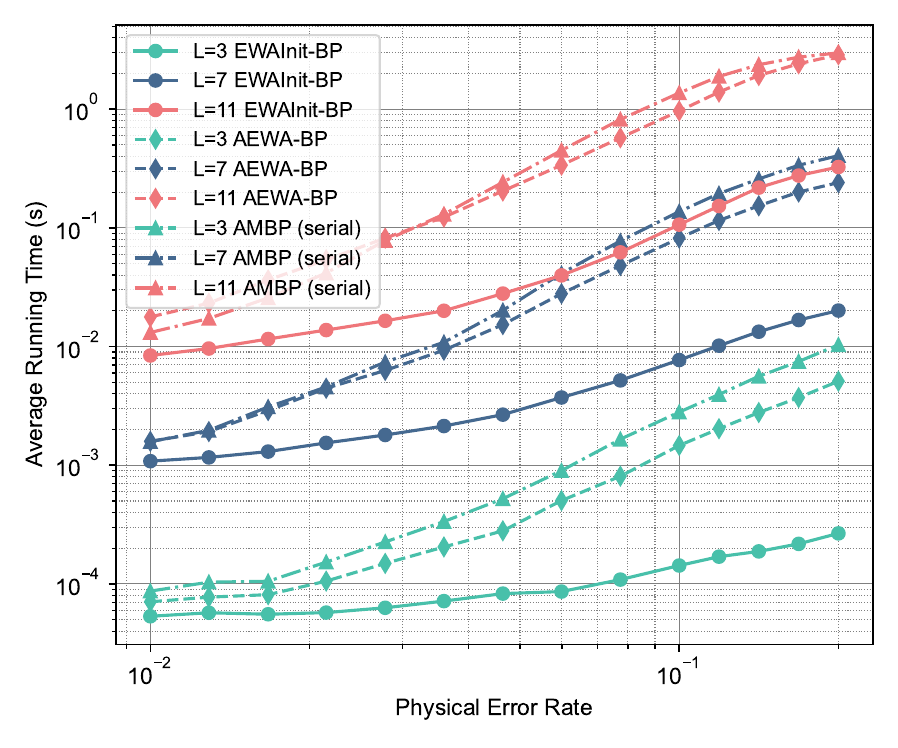}
  \caption{Average running time of EWAInit-BP, AEWA-BP, and serial AMBP on the planar surface code with L = 3, 7, 11. For the latter two algorithms that utilize adaptive methods, the running time increases significantly as the physical error rate grows.}
  \label{fig:AMBP_time}
\end{figure}

However, algorithms utilizing adaptive methods show higher scaling. As the physical error rate increases, these algorithms tend to fully traverse the parameter range, resulting in runtime scaling not only with the number of iterations but also with the parameter range length in high-error-rate scenarios. In this experiment, the parameter range is set to 11 values. It can be observed that the runtime of adaptive methods is an order of magnitude higher than that of EWAInit-BP under high error rates. Table. \ref{tb:AMBP_time} presents the runtime (ms) of the three algorithms under several error rates, offering a more intuitive perspective.

\begin{table}[!ht]
  \centering
  \caption{Comparison of the running time of EWAInit-BP, AEWA-BP, and AMBP under certain physical error rates}
  \label{tb:AMBP_time}
  \small
  \begin{tabular}{cccc}
    \toprule
    PER & EWAInit-BP (ms) & AEWA-BP (ms) & AMBP (ms) \\
    \midrule
    $10^{-2}$ & 1.1 & 1.6 & 1.6 \\
    $6\times 10^{-2}$ & 3.7 & 27.9 & 41.3 \\
    $10^{-1}$ & 7.7 & 81.9 & 137.0 \\
    $2\times 10^{-1}$ & 20.1 & 240.6 & 405.6 \\
    \bottomrule
  \end{tabular}
\end{table}

\bibliography{references}
\bibliographystyle{ieeetr}

\EOD

\end{document}